\documentclass[11pt]{article}
\usepackage{amsmath,epsfig}
\usepackage{fullpage}

\usepackage{graphicx}

\usepackage{array}

\usepackage{amsmath,amsthm,amssymb}
\usepackage[ruled,vlined,linesnumbered]{algorithm2e}
\usepackage{qtree}
\usepackage{subfigure}
\usepackage{url}

\usepackage{multirow}
\usepackage{hhline}
\usepackage{booktabs}

\usepackage[usenames,dvipsnames]{pstricks}
\usepackage{epsfig}
\usepackage{pst-grad} 
\usepackage{pst-plot} 

\newtheorem{theorem}{Theorem}[section]
\newtheorem{corollary}{Corollary}[section]
\newtheorem{lemma}{Lemma}[section]
\newtheorem{problem}{Problem}
\newtheorem{definition}{Definition}[section]
\newtheorem{observation}{Observation}[section]

\newtheorem{fact}{Fact}[section]

\newcommand{\qedsymb}{\hfill{\rule{2mm}{2mm}}}

\DeclareMathOperator{\nexthit}{NextHit}
\DeclareMathOperator{\ds}{DirectSample}

\def\0{\phantom{0}}

\newcommand{\st}{:}

\newcommand{\remove}[1]{}

\newbox\tallstrutbox
\setbox\tallstrutbox=\hbox{\vrule height10pt depth 3.5pt width0pt}
\def\tallstrut{\relax\ifmmode\copy\tallstutbox\else\unhcopy\tallstrutbox\fi}

\newbox\tallerstrutbox
\setbox\tallerstrutbox=\hbox{\vrule height14pt depth 3.5pt width0pt}
\def\tallerstrut{\relax\ifmmode\copy\tallerstutbox\else\unhcopy\tallerstrutbox\fi}

\DeclareMathOperator{\e}{E}
\DeclareMathOperator{\var}{Var}
\DeclareMathOperator{\samplesize}{\lceil 60/\epsilon^2\rceil}
\DeclareMathOperator{\maxlevel}{\lfloor \log p \rfloor}

\title{Boosting the Basic Counting on Distributed Streams}

\date{}

\author{%
  {Bojian Xu
  }
  \vspace{1.6mm}\\
  Department of Computer Science\\ Eastern Washington University,
  Cheney, WA 99004, U.S.A.\\
  bojianxu@ewu.edu }
\begin{document}
\maketitle

\begin{abstract} 
  We revisit the classic basic counting problem in the distributed
  streaming model that was studied by Gibbons and Tirthapura
  (GT). 
%
%
  In the solution for maintaining an $(\epsilon,\delta)$-estimate, as
  what GT's method does, we make the following new contributions:
  (1) For a bit stream of size $n$, where each bit has a probability
  at least $\gamma$ to be 1, we exponentially reduced the
  average total processing time from GT's $\Theta(n \log(1/\delta))$ to
  $O((1/(\gamma\epsilon^2))(\log^2 n) \log(1/\delta))$, thus providing
  the first sublinear-time streaming algorithm for this problem.
  (2) In addition to an overall much faster processing speed, our
  method provides a new tradeoff that a lower accuracy
  demand (a larger value for $\epsilon$) promises a faster processing
  speed, whereas GT's processing speed is $\Theta(n \log(1/\delta))$
  in any case and for any $\epsilon$.
  (3) The worst-case total time cost of our method matches GT's
  $\Theta(n\log(1/\delta))$, which is necessary but rarely occurs
  in our method.
  (4) The space usage overhead in our method is a lower order term
  compared with GT's space usage and  occurs only $O(\log n)$ times
  during the 
  stream processing and is too negligible to be detected by the
  operating system in
  practice. 
  We further validate these solid theoretical results with
  experiments on both real-world and synthetic data, showing that our
  method is faster than GT's by a factor of several to
  several thousands depending on the stream size and accuracy demands,
  without any detectable space usage overhead.
  Our method is based on a faster sampling technique that we design
  for boosting GT's method and we believe this technique can be of
  other 
  interest.
\end{abstract}

%
\section{Introduction}






Advances in modern science and technology have given rise to massive
data (or so-called big data). Some of the data naturally arrives as
{\em streams}. Examples include  network data packets passing through
a router, environmental data collected by sensor networks, and search
requests received by search engines. In many cases, such massive
streaming data needs to be monitored in a real-time fashion. Such data
process requirements make conventional methods such as storing them in
a relational database and issuing SQL queries thereafter infeasible,
and thus brings up the phenomenon of
 \emph{data stream processing}~\cite{muthu-book,BBDMW02}. In data stream
processing, the workspace is often orders of magnitude smaller than
the stream size, requiring the data be processed in one pass.

However, most streaming algorithms need to look at every data element
at least
once~\cite{AMS99,GKMS01,AGMS99,MG82,CGMR05,CMZ06,CMY08,SBAS04,DGIM02-SICOMP,BJKST02}
(see~\cite{muthu-book,BBDMW02} for many other example references).  In some
cases where extremely fast paced streaming data is involved, 
even a single glance at every stream element can be unaffordable. For
example, a typical OC48 link transfers 2.5 Gbits per second and AT\&T
backbone networks carry over 15 petabytes of data traffic on an
average business day. Deploying a streaming algorithm 
for monitoring purpose to  process every data element
in such massive data streams is very computationally expensive and can greatly hurt the 
performance of the system. 
The goal of sublinear-time algorithms is to solve computational
problems without having to look at every input data element. However,
in sublinear time algorithm design, the input data is often stored
statically~\cite{Fischer01theart,Goldreich98,DanaRon08-survey,RS-SIDMA11},
meaning we can visit any part of the input data at any time if needed.

In this paper, we demonstrate that designing a sublinear-time
algorithm for streaming data is also possible, without losing accuracy
guarantee compared with linear-time streaming algorithms. In
particular, we proposed the first streaming algorithm for the
distributed basic counting problem using time sublinear of the stream
size in the average case.  To our best knowledge, the best prior
result~\cite{GT01} for solving this problem has to visit every stream
element at least once and thus needs a time cost at least linear of
the stream size in any case.

\paragraph{Distributed basic counting.} 
Alice and Bob (called processors) are processing two geographically
distributed bit streams $A=\{a_1, \ldots, a_n\}$ and $B=\{b_1, \ldots,
b_n\}$, respectively and in parallel. In each stream, the $i$th bit is
received and processed before the $j$th bit, if $i<j$.  Upon receiving
a query, the referee, who is located on another remote site, wants to
know the number of 1-bits in the bit-wise {\tt OR} of the two bit
streams that Alice and Bob have observed and processed. 
$$
U(A,B)=\sum_{i=1}^{n}(a_i \lor b_i)
$$
where $\lor$ is the bit-wise logical {\tt OR} operator and $n$ is the
number of bits that Alice and Bob have both received when the query
arrives at the referee.  Note that both streams evolve 
over time and thus the stream size $n$ and the value of $U$ monotonically
increase over time.
The constraints and challenges in the computation of $U$ are: (1) no
direct communication between processors is allowed since there are no
direct connection between Alice and Bob, (2) use small workspace on the
processors as well as on the referee, and (3) use small communication
cost on the links connecting the processors and the referee.

The problem can be generalized to $k$ streams $R_1, R_2, \ldots, R_k$,
processed by $k$ processors respectively, for some constant $k\geq 2$.
For $j=1,2,\ldots, k$, we write the stream $R_j$ in the form of
$\{r_{j,1}, r_{j,2}, \ldots, r_{j,n}\}$.  Upon receiving a query,
the referee wants to know the number of 1-bits in the bit-wise {\tt
  OR} of the $k$ streams:
\begin{equation}
\label{eqn:uk}
U(R_1,R_2,\ldots, R_k)=\sum_{i=1}^{n}(r_{1,i} \lor r_{2,i} \lor \ldots
\lor r_{k,i})
\end{equation}
The same constraints and challenges for the 2-stream case
hold in this general setting. Figure~\ref{fig:dc} shows
the system setting that is assumed in the distributed basic
counting. Because our method for the 2-stream case can be
easily extended for the general setting, we will focus 
on the 2-stream case in our presentation. The extension for 
the general setting will be presented in the end. 

We refer readers to~\cite{GT01} for a detailed discussion on the
extensive applications of the distributed basic counting
in large-scale data aggregation and monitoring.

\begin{figure}
  \centering
  \includegraphics[scale=0.4]{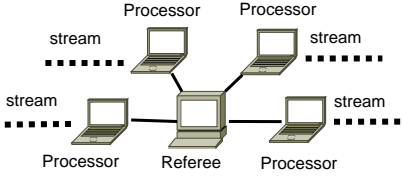}
  \caption{The setting in distributed basic counting}
  \label{fig:dc}
\end{figure}

\paragraph{Prior work.} 
A naive method for the referee to maintain the knowledge of $U$ is to
get Alice and Bob to continuously forward their stream elements to the
referee. The referee will then simply do a straightforward calculation
of $U$ in one pass of the two streams, using $O(\log n)$ bits of
workspace at both the processors and the referee. However, this
approach introduces a high communication cost between the processors
and the referee, which is prohibited in many applications such as
network monitoring. To reduce the communication cost, Gibbons and
Tirthapura proposed a communication-efficient distributed computing
scheme~\cite{GT01}, where Alice and Bob each maintains a small-space
data structure (a.k.a.\ sketch) without communicating neither to each
other nor to the referee over the course of stream processing.  When
the query arrives, the referee will first notify Alice and Bob to send
their sketches to the referee. The referee will then retrieve the
knowledge of $U$ from the sketches. However, under this distributed
computing setting, $\Omega(n)$ bits of communication cost is necessary
to get the exact value of $U$ even for randomized
algorithms~\cite{KN97}.  It is also shown that $\Omega(\sqrt{n})$ bits
of workspace is necessary at each processor even for an approximate
answer for $U$ with a relative error bound, if Alice and Bob sample
their streams \emph{independently}~\cite{GT01}. In order to achieve a
solution of both workspace and communication efficiency, Gibbons and
Tirthapura proposed the \emph{coordinated adaptive sampling} (a.k.a.\
distinct sampling) technique that uses only
$O((1/\epsilon^2)\log(1/\delta)\log n)$ bits of workspace at each
processor and the referee and $O((1/\epsilon^2)\log(1/\delta)\log n)$
bits of communication cost per link and per query. By using these
sublinear (of stream size) space and communication cost, their
technique guarantees an \emph{$(\epsilon,\delta)$-estimate} for
$U$~\cite{GT01}.  Their algorithm can be trivially extended to provide
an $(\epsilon,\delta)$-estimate of $U$ over multiple streams with the
same aforementioned workspace and communication cost.

\begin{definition}
  Given the parameters $\epsilon$ and $\delta$, $0 <
  \epsilon, \delta < 1$, the $(\epsilon,\delta)$-estimate of
  a nonnegative variable $X$ is a random variable $\hat{X}$, such that:
$$
\Pr\left[|\hat{X}-X| \leq \epsilon X\right] \geq 1-\delta
$$
\end{definition}
In particular, the $(\epsilon,0)$-estimate is also called
\emph{$\epsilon$-estimate}.

\begin{definition}[$\gamma$-random bit stream.] 
\label{def:randstream}
A bit stream $\{a_1,
  a_2, \ldots, a_n\}$ is a \emph{$\gamma$-random bit stream}, if all
  the bits in the stream are mutually independent and $\Pr[a_i =
  1] \geq \gamma$, for $i=1,2,\ldots, n$, where $0<\gamma<1$ is a
  constant.
\end{definition}

The notion of $\gamma$-random bit stream can (roughly) capture the distribution of
many real-world bit streams.

\subsection{Our contribution}
We designed a novel sampling technique that enables us to sample the
stream without having to check every stream element.  By using this
fast sampling technique, we are able to boost GT's processing speed in
maintaining an $(\epsilon,\delta)$-estimate of the distributed basic
counting with negligible extra space usage. Table~\ref{tab:compare}
summarizes the performance comparison of our method and GT's.

\begin{itemize}
\item The average total processing time for a $\gamma$-random bit stream
is reduced from GT's $\Theta\bigl(n \log\frac 1 \delta\bigr)$ to
$O\bigl(\frac{1}{\gamma\epsilon^2}\log^2 n \log\frac 1 \delta\bigr)$.  Our method not
only \emph{exponentially} improves the overall processing speed,
providing the first sublinear-time algorithm in the average case, but
also provides a new tradeoff that a lower accuracy demand
(a larger value for $\epsilon$) promises a faster processing speed,
whereas GT's method spends $\Theta\bigl(n \log\frac 1 \delta\bigr)$ time regardless
of the accuracy demand.

\item Our method's worst-case total processing time matches GT's
$\Theta(n\log(1/\delta))$, which is necessary. However, this
worst-case time cost rarely occurs with our method, whereas GT's
method always needs $\Theta(n\log(1/\delta))$ time in any case.

\item  Each processor uses $O((1/\epsilon^2 + \log
n)\log(1/\delta) \log n)$ bits of workspace in our method. 
Compared with GT's space usage of
$O((1/\epsilon^2)\log(1/\delta) \log n)$
bits, our method's extra space usage 
is a lower-order term for any
real-world bit stream and a reasonably small $\epsilon$ (say
$\epsilon \leq 0.1$) and is indeed undetectable in our experiments with
both real-world and synthetic data. Further, this extra space cost
occurs only $O(\log n)$ times in average during the stream processing.

\item The workspace at the referee and the communication cost between the
processors and the referee in our method remains the same as GT's. 

\item  We conducted a comprehensive experimental study using both
real-world and synthetic data. All experimental results show that our
method is faster than GT's by a factor of several to several
thousands depending on the stream size and accuracy demand. 
Our method can potentially save the vast majority of the
processing time and energy that is consumed by GT's method in the
real-world deployment, where the stream size can be
nearly unbounded. All experiments also show that the OS
does not detect any extra space cost used by our method compared with
GT's.

\item The fast sampling technique we proposed can also be of other
independent interest in the design of sampling-based algorithms 
solving other problems.
\end{itemize}

\begin{table}
\label{tab:compare}
  \centering

{\footnotesize
  \begin{tabular}{m{4.3cm}|l|m{4cm}|m{4cm}}
\hline
\tallstrut
& \textrm{GT's~\cite{GT01}} & \textrm{Ours} & \textrm{Note}\\
\hline
\hline
\tallstrut
\textrm{Worst-case total time cost} & $\Theta\bigl(n\log
\frac{1}{\delta}\bigr)$ & $\Theta\bigl(n\log \frac{1}{\delta}\bigr)$ &
always occurs with GT's, but rarely occurs with ours.\\
\tallstrut
\textrm{Worst-case per-item time cost} & $\Theta\bigl(\log n \log \frac 1 \delta\bigr)$ & $\Theta\bigl(\log n \log \frac 1 \delta\bigr)$ &  \\ 
\tallstrut
\textrm{Avg.-case total time cost for a $\gamma$-random bit stream} &  $\Theta\bigl(n\log \frac{1}{\delta}\bigr)$ & $O\bigl(\frac{1}{\gamma\epsilon^2}\log^2 n \log \frac 1 \delta\bigr)$ & \textrm{significantly improved} \\
\tallstrut
\textrm{Avg.-case per-item time cost for a $\gamma$-random bit stream} & $\Theta\bigl(\log \frac{1}{\delta}\bigr)$ & $O\bigl(\frac{1}{n\gamma\epsilon^2}\log^2 n \log \frac 1 \delta\bigr) =o(1)$, when $n$ is large. & \textrm{significantly improved} \\ 
\hline
\textrm{\tallstrut Space cost per-processor (bits)} & $O\bigl(\frac{1}{\epsilon^2}\log \frac 1\delta\log n\bigr)$ & $O\bigl((\frac{1}{\epsilon^2} + \log n) \log \frac 1\delta \log n\bigr)$ & \textrm{negligible overhead} \\
\tallstrut
\textrm{Space cost by the referee (bits)} & $O\bigl(\frac{1}{\epsilon^2}\log \frac 1\delta\log n\bigr)$ & $O\bigl(\frac{1}{\epsilon^2}\log \frac 1\delta\log n\bigr)$ & \\
\hline
\textrm{\tallstrut Comm.\ cost per query (bits)}& $O\bigl(\frac{1}{\epsilon^2}\log \frac 1\delta\log n\bigr)$ & $O\bigl(\frac{1}{\epsilon^2}\log \frac 1\delta\log n\bigr)$ & \\
\hline
  \end{tabular}
}
\caption{The performance comparison between GT's
  method and ours.}
\end{table}

\subsection{Paper organization} 
After a survey of related work, we will present a high-level overview
of our method in Section~\ref{sec:overview}, where we will introduce
the structure of the coordinated adaptive sampling and show the
opportunity for boosting.  We will then introduce the new sampling
technique for boosting in Section~\ref{sec:directsample}. By plugging
the faster sampling technique into the coordinated adaptive sampling,
we are able to present the complete picture of our method in
Section~\ref{sec:new}. The details of a comprehensive experimental
study using both real-world and synthetic data are given in
Section~\ref{sec:exp}.  This paper is concluded by
Section~\ref{sec:con}.

\section{Related work}
In this section, we summarize the results on basic counting in the
streaming model under various constraints and settings. For a broader
overview of stream processing, we refer readers to the
surveys~\cite{muthu-book,BBDMW02}. 

While the small-space basic counting of a whole single stream is
trivial, it becomes much harder when counting the number of 1-bits
on the union of multiple
geographically distributed streams. It is shown in~\cite{KN97} that
$\Omega(n)$ bits of workspace is necessary for an exact answer even for
randomized algorithms, where $n$ is the size of each single stream.
Datar et al.~\cite{DGIM02-SICOMP} considered the basic counting on a single
stream, given the constraint of sliding windows. It is easy to show
$\Omega(N)$ bits is required to maintain the exact knowledge of how many
1-bits in a sliding window of size $N$, representing the most recently received 
 $N$ bits.
They further showed a space lower bound of $\Omega((1/\epsilon)\log^2
(\epsilon N))$ bits for both deterministic and randomized algorithms
for maintaining an $\epsilon$-estimate of the basic counting over a
sliding window of size $N$. They also proposed an exponential
histogram based deterministic algorithm that guarantees an
$\epsilon$-estimate using workspace matching the above space lower
bound, but their algorithm cannot work for the union of multiple bit
streams under the sliding windows model.

Gibbons and Tirthapura~\cite{GT02} solved this new challenge by
extending and applying their technique from~\cite{GT01} to the sliding
window setting over multiple streams. Instead of maintaining one
sample as was done in~\cite{GT01} for the case without sliding window,
they maintain $\log N$~\footnote{In this paper, we use the convention
  that, unless specified explicitly, the base of logarithm function is
  $2$.} samples where every sample has a different sample probability
from $1,1/2,\ldots$, $1/(2^{\log N})$, because the number of 1-bits in
the sliding window can vary over time depending upon the input stream
and thus one fixed good sample probability cannot be predetermined
beforehand. By maintaining multiple samples, their
technique is able to pick the best samples with the right sample
probabilities from the multiple processors at the query time, so that the referee
is guaranteed to have a good estimate for the basic counting on the union of the
bit streams over the sliding window.

Xu et al.~\cite{XTB08} considered the sliding window based basic
counting on an asynchronous stream, where the arrival order of the stream
elements is not necessarily the same as the order at which they
were created. Such an asynchronous stream model is motivated by
the real-world scenarios such as  network data packets being
received out of order at the destination due to the network delay
and multi-path routing. The core idea of their solution is mostly
identical to the one from~\cite{GT02} but is modified for 
asynchronous streams. It is also easy to extend their solution
so that it can work for the union of multiple streams.  
Busch and Tirthapura~\cite{BT07} later also solved the asynchronous
stream basic counting problem over sliding windows for one stream.
Their solution is deterministic and is based on a novel data structure
called splittable histogram, but it is not clear how to extend their
method to multiple streams.

The coordinated random sampling technique~\cite{GT01} can
also be used for counting $F_0$~\cite{AMS99,BJKST02,IW03}, the number
of distinct elements, over one or the union of multiple data streams,
with or without sliding windows~\cite{GT01,GT02,Gi01vldb}.
Pavan and Tirthapura~\cite{PT07} generalized the data stream model in
the calculation for $F_0$. In their stream model, every stream element
is no longer a single integer but is a range of continuous
integers. The $F_0$ is defined as the number of distinct integers in
the union of the ranges that have been received in the stream. A
trivial solution is to expand each range into a sequence of integers
and use an existing $F_0$ algorithm to process each integer.  The
time cost for processing a range will then be at least linear of the
range size, which is not acceptable when the range size is large. They
proposed a divide and conquer based strategy, such that the time cost
for processing each range is only a logarithm of the range size. Part
of the idea behind our new sampling technique presented in this paper is
inspired by their work.

All these related work need to observe each stream element at
least once, leading to their time costs to be at least linear of the
stream size.

\section{A high-level overview}
\label{sec:overview}


The high-level structure of our method is the coordinated adaptive
sampling by Gibbon and Tirthapura~\cite{GT01}, but uses a different
hash function for the random sampling. We exploit 
the properties of the hash function, so that we can do random
sampling over the data stream without having to check every stream
element.

\subsection{Coordinated adaptive random sampling}
Alice and Bob use the same sampling procedure and the sample size, so
we will only describe the behavior of Alice.

{\bf Random sampling.}
Alice maintains a sample of some known size $\alpha$, which will be
determined later. She randomly selects each 1-bit (by storing the 
bit's stream location index) into the sample with some
probability $p$. After processing the stream, the number of 1-bits
selected into the sample multiplied by $1/p$ can be a good estimate of
the number of 1-bits in the stream, if $\alpha$ is large enough.

{\bf Adaptive random sampling.} However, Alice does not know in
advance how many 1-bits will be present in her stream and thus cannot
decide an appropriate sample probability $p$. If the sample
probability is too high, the sample size may not be big enough to accommodate
all the selected 1-bits; if the sample probability is too small,
the sample may not select enough 1-bits to yield a good estimate.
To overcome this difficulty, Alice
\emph{adaptively} changes the sample probability over the course of
her stream processing.
The sample probability is determined by the \emph{sample
  level} $\ell$, which starts from $0$ and increases towards  $1, 2, \ldots$. At sample level
$\ell$, every incoming 1-bit is selected into the sample with sample
probability $P_\ell =1/2^\ell$ (In our method, $P_\ell$ is not
exactly but is nearly equal to $1/2^\ell$ and will be clear later).
Alice is always aware of her current sample probability by
remembering her current sample level $\ell$.

{\bf Coordinated adaptive random sampling.}  However, Gibbons and
Tirthapura showed that if all the processors do their random sampling
\emph{independently}, $\Omega(\sqrt{n})$ bits of workspace at each
processor is necessary even for an estimate of the distributed basic
counting with a relative error bound~\cite{GT01}.  The way
in~\cite{GT01} to overcome this space lower bound is to
\emph{coordinate} the random sampling procedures by using a common
hash function to simulate the sampling procedure.  The hash function
used in~\cite{GT01} is pair-wise independent and is defined over the
field $GF(2^m)$, where $m=\log n$. Their hash function maps the stream
location indexes $\{1,2,\ldots,n\}$ to the sample levels $\{0, 1,
\ldots, m\}$. The hash function provides that the probability that a
stream location index is hashed to a particular sample level $i$ is
exactly equal to $1/2^{i+1}$.  An incoming 1-bit will be selected into
the sample if and only if its stream location index is hashed to a
sample level $i\geq \ell$, where $\ell$ is the current sample
level. Thus, in average, only $1/2^{\ell+1}+1/2^{\ell+2}+\ldots +
1/2^m \approx 1/2^\ell$ of the $n$ stream element locations will be
selected, which is the goal of the adaptive sampling --- each stream
location should be selected with probability $P_\ell=1/2^\ell$.
Clearly, coordinated adaptive random sampling yields the same set of
selected stream locations at all processors if they are on the same
sample level.  Of course, those selected locations that do not have
1-bits will not be saved in the sample.

{\bf When the sample is full.}  
The sample level will be incremented by one, when the sample becomes
full. All the 1-bits that are currently in sample and whose hash
values are less than the new sample level $\ell$ will be discarded
from the sample. By doing so, all the 1-bits that Alice has received
so far have been assigned the same sample probability $P_\ell$ for
selection. Alice will then continue her processing of the new incoming
1-bits using the current sample probability $P_\ell$ until the sample
becomes full again.

{\bf When the referee receives a query for $U$.}  All processors will
send their samples and sample levels to the referee. For each sample,
if its sample level is smaller than $\ell_{max}$, the largest sample
level the referee has received, the referee will increase its sample
level to $\ell_{max}$ by retaining only those 1-bits whose hash values
are not less than $\ell_{max}$. After this resampling step,
all the samples will share the same sample level. The referee will
then union all the samples by conducting bit-wise {\tt OR} of the
selected 1-bits according to their stream locations. The number of
1-bits in union multiplied by $1/P_{\ell_{max}}$ will be returned by the referee as the
answer to the query for $U$.

\subsection{Use a different hash function for sampling}
\label{subsec:hash}
Our method follows GT's structure but 
uses a different hash function to conduct the coordinated adaptive
sampling procedure. Since the choice of the hash function is critical for the
improvement of the processing speed, we provide its details in the
following. We first pick a prime number $p$ uniformly at random from
$[10n,20n]$, then pick another two numbers $a$ and $b$ from
$\{0,1,\ldots, p-1\}$ and $a\neq 0$. The hash function $h:\{1,2,\ldots,
n\}\rightarrow \{0,1,\ldots, p-1\}$ is defined as: $h(x)=(a\cdot x +
b) \mod p $. It is well known~\cite{CW79} that: 
\begin{itemize}
\item 
$h(x)$ is
uniformly distributed in $\{0,1,\ldots, p-1\}$. 
\item 
$h$ is
pair-wise independent, i.e., for any $x_1\neq x_2$ and $y_1,y_2$:
$\Pr[(h(x_1)=y_1) \land (h(x_2)=y_2)] = \Pr[h(x_1)=y_1]\cdot
\Pr[h(x_2)=y_2]$.
\end{itemize}

For each
$\ell \in \{0,1,\ldots, \maxlevel\}$,  we define:
$$
R_\ell = \left\{0,1,\ldots,\lfloor p/2^\ell \rfloor -1\right\}
$$
Our sampling policy is that a stream element $a_i$ will be selected
into the sample if $a_i=1$ and $h(i)\in R_\ell$. Since the hash values
$\{0,1,\ldots, p-1\}$ are
uniformly distributed, the probability of
selecting any 1-bit at sample level $\ell$ is:
$$
P_\ell = |R_\ell|/p=\lfloor p/2^\ell \rfloor/p \approx 1/2^\ell
$$

Algo.~\ref{algo:gt} shows the pseudocode of GT's coordinated
adaptive sampling, plugged in our new hash function, for estimating
$U$ of two streams.

\begin{algorithm}[t]
\small
  \caption{GT's coordinated adaptive sampling for distributed basic
    counting using the pairwise independent hash function $h(x)=(ax+b)\mod p$.}
\label{algo:gt}

\KwIn{Two geographically distributed bit streams $A=\{a_1, a_2, \ldots,
  a_n\}$ and $B=\{b_1, b_2, \ldots, b_n\}$, processed in parallel
  by Alice and Bob, respectively.} 

\KwOut{$U(A,B) = \sum_{i=1}^n (a_i \lor b_i)$, returned by the referee.}

\smallskip

{\bf Randomly pick a pairwise independent hash function:} $h(x) = (ax+b) \mod p$,
where $p$ is a prime number randomly picked from $[10n,20n]$, and
$a,b$ are randomly picked from $\{0,1,\ldots, p-1\}$ and $a\neq 0$.

\smallskip

\underline{\bf Alice}: 
\smallskip 

Initialize an empty sample $S$ of size $\alpha$;
$\ell\leftarrow 0$\;

\For{$i=1, 2, \ldots$ n\label{line:alicefor}}{
\If{$a_i=1$ and $h(i) \in R_\ell$ \label{line:alice}}
{
 $S \leftarrow S \cup \{i\}$\;
\While (\tcp*[f]{Sample is full.}){$|S| > \alpha$}{
  $\ell = \ell + 1$\;
  \lIf{$\ell > \maxlevel$}{Exit} \label{line:fail} 
  \tcp*{Algorithm fails.}
  Discard every $x\in S$ such that $h(x) \notin R_\ell$\;
}
}
}

\smallskip

\underline{\bf Bob}: 
(Symmetric to Alice)

\smallskip

\underline{\bf Referee}: \tcp*[f]{Upon receiving a query for $U$.}

\smallskip 

Receive $S_A,\ell_A$ from Alice and $S_B,\ell_B$ from Bob;

$\ell^*\leftarrow\max(\ell_A, \ell_B)$;
\label{line:l}

\lIf{$\ell_A < \ell^*$}{
    Discard every $x\in S_A$ such that 
        $h(x) \notin R_{\ell^*}$\;
}
\lElseIf{$\ell_B < \ell^*$}{
     Discard every $x\in S_B$ such that 
        $h(x) \notin R_{\ell^*}$\;
}
\Return{$|S_A \cup S_B|/ P_{\ell^*}$}; 
\end{algorithm}

\subsection{The opportunity for speedup}
GT's method checks every stream element at least once
(Line~\ref{line:alicefor}--\ref{line:alice}, Algo.~\ref{algo:gt}),
which yields their time cost for processing a stream of $n$ bits is at
least $\Theta(n)$.
We observe that the hash values $h(1), h(2), \ldots$ of the hash
function $h$ we use (Section~\ref{subsec:hash}) follows some pattern,
which will be made clear later. By taking advantage of this pattern,
it is possible to skip over some stream elements without checking
their hash values, because their hash values can be proved to be out
of $R_\ell$. In particular, suppose $a_i$ is the element that we are
currently processing and $a_{i+d}$ is the next element whose hash
value is within $R_\ell$, we are able to design an algorithm that
finds the value of $d$ using $O(\log d)$ time. That says we will not
need to literally check the elements $a_{i+1}, a_{i+2}, \ldots,
a_{a+d-1}$ one by one.  Our algorithm for finding the value of $d$ is
based on a recursive decomposition of the task, so that the new
instance of the problem has a significantly smaller problem size. The
recursive decomposition is based some properties of the hash function
$h(x)=(a\cdot x+b)\mod p$ that we use for coordinated adaptive
sampling.

\section{The new technique: direct sampling}
\label{sec:directsample}


The challenge for fast sampling is to find the next stream location
that will be selected into the sample if that location has a 1-bit,
without having to check the hash value of every stream location up to
that sampled location. That is, suppose we are now at sample level
$\ell$ and at the stream location $x$, we want to compute the
following function quickly:
$$
\ds(x,\ell, p, a, b)=x+\mathcal{N}_x^\ell
$$
where
$$
\mathcal{N}_x^\ell = \min\{i \geq 0 \mid h(x+i) = (a(x+i)+b)\mod p 
\in R_\ell\}
$$
Note that for an arbitrary setting of the parameters of $h$,
$\mathcal{N}_x^\ell$ may not be well defined because the set $\{i \geq
0 \mid h(x+i) \in R_\ell\}$ can be empty.  For
example, if $p=8, a=4, b=2, \ell=2, x=1$, then for any $i\geq 0$, the
value of $h(x+i)$ is always either $2$ or $6$, neither of which 
belongs to $R_\ell = \{0,1\}$.
However, it can be shown that when $p$ is a prime number, 
$\mathcal{N}_x^\ell$ is always well defined. We first prove that $\{h(x),
h(x+1), \ldots h(x+p-1)\}$ and $\{0,1,\ldots, p-1\}$ are actually the same set
of $p$ distinct numbers.
 
\begin{lemma}
 \label{lem:mapping}
 For any $p>0$, $0<a<p$, $0\leq b<p$, and $x\geq 1$, if $p$ is a prime
 number, then:
$$
\{h(x), h(x+1), \ldots h(x+p-1)\}
=\{0,1,\ldots, p-1\}
$$
\end{lemma}

\begin{proof}
We prove the lemma by contradiction. 
Let $A$ denote $\{h(x), h(x+1), \ldots h(x+p-1)\}$ and B denote 
 $\{0,1,\ldots, p-1\}$. 
Suppose $A\neq B$, then
since every $h(x+k)$ for $k=0,1,\ldots, p-1$
is a member of $B$ and $|A|=|B|$, there  must exist two integers $i$ and $j$, such
that $0\leq i < j \leq p-1$ and $h(x+i) = h(x+j)$. That is, 
\begin{eqnarray}
&&a(x+i)+b = a(x+j)+b \mod p\nonumber \\
&\Longleftrightarrow& ai = aj \mod p \nonumber \\
&\Longleftrightarrow& a(j-i) = 0 \mod p 
\label{eqn:primefactor}
\end{eqnarray}
Combining the fact that $0< a < p$ and $0< j-i < p$, the above
Equation~\ref{eqn:primefactor}
indicates that $p$ is not a prime number,
which is a contradiction, thus the lemma is proved.
\end{proof}

\begin{lemma}
\label{lem:exist}
For any $p>0$, $0<a<p$, $0\leq b<p$, $x\geq 1$, and $L \geq 0$, if $p$
is a prime number, then $\{i\geq 0 \mid h(x+i) \leq L\}$ is always
not empty.
\end{lemma}

\begin{proof}
  By Lemma~\ref{lem:mapping}, we know that $\{h(x),h(x+1),\ldots,$
  $h(x+p-1)\}=\{0,1,\ldots, p-1\}$. Because $L\geq 0$, there must exist
  at least one member in the set $\{h(x),h(x+1),\ldots, h(x+p-1)\}$
  whose value is not larger than $L$, so the lemma is proved.
\end{proof}

\begin{lemma}
\label{lem:exist2}
$\mathcal{N}_x^\ell$ always exists.
\end{lemma}

\begin{proof}
  Recall that by design the sample level $\ell$ is no more than 
  $\maxlevel$ (Line~\ref{line:fail}, Algo.~\ref{algo:gt}), so
  $\lfloor 2^{-\ell}p \rfloor-1 \geq 0$ is always true.  
By setting $L=\lfloor 2^{-\ell}p
  \rfloor-1 \geq 0$, Lemma~\ref{lem:exist} has proved that
the set $\{i\geq 0 \mid h(x+i) \leq L\}$ is always not empty, 
so 
$\mathcal{N}_x^\ell 
= \min\{i \geq 0 \mid h(x+i) \in R_\ell\}
= \min\{i \geq 0 \mid h(x+i) \leq \lfloor
2^{-\ell}p\rfloor -1\}$ 
is always well defined and exists. 
\end{proof}

We can now focus on the design of an efficient algorithm 
for finding $\mathcal{N}_x^\ell$.

\begin{problem}
\label{prob:nexthit}
  Given integers $p > 0$, $0\leq a<p$, $0\leq u<p$, and $L\geq 0$,
  computer the following $d$: 
$$
d = \left\{
\begin{array}{l}
  \min\{i\geq 0 \mid (u+i\cdot a)\mod p \leq L \},\\
   \textit{\ \ \ \ \ \ \  \ \ \ if\ \ }
  \{i\geq 0 \mid (u+i\cdot a)\mod p \leq L \} \neq \emptyset\\
-1, \textit{\ \ \ \ \ otherwise}
\end{array}
\right .
$$
\end{problem}

Let $Z_p$ denote the ring of the nonnegative numbers modulo
$p$. Observe that the sequence of the values of $h(x+i)$, for
$i=0,1,2,\ldots$, is an arithmetic progression over $Z_p$ with 
common difference $a$. The task of finding $\mathcal{N}_x^\ell$ is
reduced to finding $d$ ($=\mathcal{N}_x^\ell$) in 
Problem~\ref{prob:nexthit} by setting:
$$u=h(x) = (ax+b) \mod p, \ \ \  L=\lfloor 2^{-\ell} p \rfloor - 1$$

Let $S=\langle u\mod p, (u+a)\mod p, (u+2a)\mod p, \ldots\rangle$.
For $i=0,1,\ldots$, let
$S[i]=(u+i\cdot a)\mod p$, the $i$th element in the sequence $S$.
Problem~\ref{prob:nexthit} can be restated as follows: find the
smallest $i\geq 0$, such that $S[i]\leq L$, or report that
such $i$ does not exist.

\subsection{Possible solutions.}
A naive method, as  used in~\cite{GT01}, is to iteratively check the
value of $(u+i\cdot a)\mod p$ for $i=0,1,\ldots$. It takes $O(d)$ time
if a nonnegative $d$ is eventually found and the procedure will even
never stop if a nonnegative $d$ does not exist at all.

A better solution is to use the {\tt Hits} function
from~\cite{PT07}. Given a fixed-sized prefix of the sequence $S$, {\tt
  Hits} can efficiently calculate the number of elements in the prefix
whose values are less than or equal to $L$ using $O(\log y)$ time,
where $y$ is the length of the prefix. Assuming a nonnegative $d$
exists for Problem~\ref{prob:nexthit}, 
if we can make a good guess of the length of the prefix, such
that $y\geq d$, meaning the prefix includes at least one element whose
value is less than or equal to $L$. Then we can use a binary search
over the prefix to locate the first element whose value is less than
or equal to $L$. Altogether, it will take $O(\log^2 y)$ time, because
there are $O(\log y)$ binary search steps and each binary search step
takes $O(\log y)$ time for {\tt Hits}. However, this method has
several unsolved issues: (1) It is not clear how to make a good guess
of the length of the prefix so that $y\geq d$, since $d$ is unknown: If
$y \gg d$, it wastes computational time; If $y< d$, we will have to
guess a longer prefix of $S$, leading to an inefficient procedure.
(2) Even if a good guess of $y$ is made, it will take
$O(\log^2 y)$  time to locate the first element whose value is less
than or equal to $L$. We look for a solution that
takes only $O(\log y)$ of time. (3) If  $d\geq 0$ does not exist, 
the binary search method of using {\tt
  Hits} fails.

Another possible solution is to use the {\tt MinHit} function
from~\cite{CTX-SICOMP09}. Given a prefix of size $y$ of the sequence
$S$, by using $O(\log y)$ of time, {\tt MinHit} can find the first
element in the prefix that is less than or equal to $L$ or return $-1$
if such element does not exist. Similar unsolved issues exist: (1) It
is not clear how to make a good guess of the size of a prefix so that
it includes at least one element that is not larger than $L$. Bad
guesses lead to wasting of computational time as we have explained in
the possible solution using {\tt Hits}. (2) If
 $d\geq 0$ does not exist, the method of using {\tt
  MinHit} fails.

\begin{figure*}[t]
\label{fig:example}
\begin{center}
{\footnotesize
\begin{tabular}{r|rr|rrr|rrr|rrrr|rrr|r}
  $i$ & $0$ & $1$ & $2$ & $3$ & $4$ & $5$ & $6$ & $7$ & $8$ & $9$ &
  $10$ & $11$ & $12$ & $13$ & $14$ & $\cdots$\\
\hline
\hline
\tallstrut $S[i]$ & ${\bf 7}$ & $11$ & ${\bf 2}$ & $6$ & $10$ & ${\bf 1}$ &
$5$ & $9$ & ${\bf 0}$ &
  $4$ & $8$ & $12$ & ${\bf 3}$ & $7$ & $11$ & $\cdots$ \\
\hline
& \tallstrut $f_0$   & & $f_1$ &  &  & $f_2$ &
 &  & $f_3$ &
 & & & $f_4$ & & & $\cdots$\\
\cline{2-17}
&\multicolumn{2}{c|}{\tallstrut$S_0$}&\multicolumn{3}{c|}{$S_1$}&\multicolumn{3}{c|}{$S_2$}&\multicolumn{4}{c|}{$S_3$}&\multicolumn{3}{c|}{$S_4$}&\multicolumn{1}{l}{$\cdots$}
\end{tabular}
}
\caption{An example sequence $S=\langle u+0\cdot a\mod p, (u+1\cdot
  a)\mod p, (u+2\cdot a)\mod p, \ldots\rangle$, where $u=7, a=4,
  p=13$. For $i=0,1,\ldots$ : (1) $S[i]$ denotes the $i$th element in
  $S$; (2) The subsequence $S_i$ denotes the $i$th round of the
  progression over $Z_p$; (3) $f_i$ is the smallest element in
  $S_i$. Suppose $L=1$ in the setting of Problem~\ref{prob:nexthit},
  the answer should be $d=5$, because $S[5]$ is the first element whose
  value is not larger than $L$.}
\end{center}
\end{figure*}

\subsection{Our solution} 
We now present our algorithm called {\tt NextHit} for
solving Problem~\ref{prob:nexthit}. Our approach is to directly modify the
internal mechanism of the {\tt MinHit} function so  it can work
with an infinitely long sequence. 

Note that the sequence $S$ is an arithmetic progression over $Z_p$
with common distance $a$.  For $i=0,1,\ldots$, let $S_i$ denote the
subsequence which is the $i$th round of the progression and $f_i$
denote the first element in $S_i$. Let $F=\langle f_0,f_1,f_2,\ldots
\rangle$ and $|S_i|$ denote the number of elements in $S_i$.
Figure~\ref{fig:example} shows an example $S$ and its 
$S_i$ subsequences and $F$ sequence.

We start the design of {\tt NextHit} with
the following critical observation from~\cite{PT07}.

\begin{observation}[Observation 2 of~\cite{PT07}]
\label{obe:first}
Sequence $\bar{F}= F\setminus\{f_0\}=\langle f_1, f_2, \ldots \rangle$
is an arithmetic progression over $Z_{a}$, with common difference
$a-r$ (or $-r$, equivalently), where $r=p\mod a$. 
That is, for every $i\geq 1$: 
\begin{equation}
\label{eqn:f}
f_i = (f_1 + (i-1)\cdot (a-r))\mod a
\end{equation}
\end{observation}

Figure~\ref{fig:example} shows an example $\bar{F}$ sequence.  

The next lemma says that if $\{i\geq 0 \mid (u+i \cdot a)\mod p \leq
L\}\neq \emptyset$, $S[d]$ must be the first element whose value is not
larger than $L$ in the sequence $F$.

\begin{lemma}[A generalization of Lemma 3.6~\cite{CTX-SICOMP09}]
\label{lem:first}
If $d\neq -1$, $S[d]=f_m\in F$, where 
$m=\min\{i\geq 0 \mid f_i\leq L\}$.
\end{lemma}

\begin{proof}
  First, we prove $S[d]\in F$. Suppose $S[d]\not\in F$ and $S[d]\in
  S_t$, for some $t$. Let $f_t=S[d']$, so $d'<d$. Since $S[d]$ and
  $f_t$ both belong to $S_t$ while $S[d]$ is not the first element of
  $S_t$, we have $f_t \leq S[d]\leq L$. Because $d'<d$, if $d'$ is not
  returned, $d$ will not be returned either. This yields a
  contradiction. Next, we prove $S[d]=f_m$. Suppose
  $S[d]=f_{m'}$, where $m'>m$. Let $f_m = S[d']$. Note that $d'<d$
  because $m<m'$. Since $d'<d$ and $S[d']\leq L$, if $d'$ is not
  returned, $d$ will not be returned either. This is also a
  contradiction.
\end{proof}

By observing $S$ is an arithmetic
progression over $Z_p$, it is easy to get the following lemma. 

\begin{lemma}[Lemma 3.7 of~\cite{CTX-SICOMP09}]
\label{lem:formular}
If $m=\min\{i\geq 0 \mid f_i\leq L\}$ exists, then $d=(mp-f_0+f_m)/a$.
\end{lemma}

\subsubsection{The overall strategy for solving Problem~\ref{prob:nexthit}} 
We will first find the value of $m$,
which will give us the value of $f_m$ due to Equation~\ref{eqn:f}.
Then, we will get the value of $d$ (thus solving
Problem~\ref{prob:nexthit}) by using Lemma~\ref{lem:formular}. That
is, the task of solving Problem~\ref{prob:nexthit} can be reduced to
the task of finding $m$, which in fact, as we will explain soon, is
another instance of Problem~\ref{prob:nexthit} with a different
parameter setting.

The next lemma will be used to convert an instance of
Problem~\ref{prob:nexthit} to
another instance of Problem~\ref{prob:nexthit} which has a different set of
parameters but returns the same answer.

\begin{lemma}[A generalization of Lemma 3.8~\cite{CTX-SICOMP09}]
\label{lem:convert}
Let $d$ denote the answer to
an instance of Problem~\ref{prob:nexthit} with parameter setting: $p,
a, u, L$.  Let  $d'$ denote the answer to another instance of
Problem~\ref{prob:nexthit} with parameter setting: $p'=p, a'=p-a,
u'=(p-u+L)\mod p, L'=L$. Then:
$d = d'$
\end{lemma}

\begin{proof}
  Let $P = \{i\geq 0 \mid (u+i\cdot a)\mod p \leq L\}$ and $P'= \{j
  \geq 0 \mid \left((p-u+L)\mod p + j\cdot (p-a)\right)\mod p \leq
  L\}$.  In the trivial case where $P$ and $P'$ are both empty, clearly
  $d=d'=-1$.  
  In the nontrivial case where $P$ and $P'$ are not both 
  empty, we first prove $P=P'$ by showing $P\subseteq P'$ and
  $P'\subseteq P$.  
\begin{enumerate}
\item
 $P \subseteq P'$.
Suppose $\gamma \in P$, then $\gamma\geq 0$ and $(u+\gamma
\cdot a)\mod p \leq L$. We want to prove $\gamma \in P'$.
\begin{eqnarray*}
&&[(p-u+L)\mod p + \gamma\cdot (p-a)]\mod p\\
&=& [p-u+L+\gamma\cdot(p-a)]\mod p\\
&=&[L-(u+\gamma \cdot a)]\mod p\\
&=& [L - (u+\gamma\cdot a)\mod p]\mod p
\leq L
\end{eqnarray*}
The inequality is due to the fact that 
$0\leq (u+\gamma\cdot a)\mod p\leq L$. So, $\gamma \in P'$.

\item $P' \subseteq P$.  Suppose $\gamma \in P'$, then $\gamma\geq 0$
  and $[(p-u+L)\mod p+\gamma \cdot (p-a)]\mod p \leq L$. We want to
  prove $\gamma \in P$.
\begin{eqnarray*}
&&[(p-u+L)\mod p + \gamma\cdot (p-a)]\mod p\\
&=& [L - (u+\gamma\cdot a)\mod p]\mod p
\leq L
\end{eqnarray*}

If $(u+\gamma\cdot a)\mod p > L$, say $(u+\gamma\cdot a)\mod
p=L+\sigma < P$ for some $\sigma>0$. From the above inequality, we can
have that $(-\sigma)\mod p=p-\sigma\leq L$, i.e., $L+\sigma \geq P$,
which yields a contradiction. So, $(u+\gamma\cdot a)\mod p
\leq L$, i.e., $\gamma \in P$.
\end{enumerate}

Because $d$ and $d'$ are the smallest values in
  $P$ and $P'$ respectively, the fact $P=P'\neq \emptyset$
  directly yields $d =d'$.
\end{proof}


\subsubsection{The mechanism of {\tt NextHit}}
Now we are ready to design the algorithmic mechanism of the {\tt NextHit}
algorithm using the discoveries we have presented.  Given an instance
of Problem~\ref{prob:nexthit} with parameter setting $p,a,u,L$, the
easy case is $u\leq L$, for which we will obviously return
$d=0$. Otherwise, because $S[d]=f_m$ (Lemma~\ref{lem:first}), we will
reduce the task of solving Problem~\ref{prob:nexthit} with parameter
setting $p,a,u,L$ to the task of finding $m$ such that $f_m$ is the
first element whose value is not larger than $L$ in the $\bar{F}$
sequence. Observe that the $\bar{F}$ sequence is also an arithmetic
progression over a smaller ring $Z_a$ with common distance $a-r$
(Observation~\ref{obe:first}), so the task of finding $m$ is actually
a new instance of Problem~\ref{prob:nexthit} with the following
parameter setting:
\begin{equation}
\label{eqn:reduction}
p_{new}=a,\ \ \  a_{new}=a-r,\ \ \  u_{new}=f_1, \ \ \  L_{new}=L
\end{equation}
Note that the returned value of the above instance actually will be
equal to $m-1$, because the $\bar{F}$ sequence is one indexing based.
After $m$ is recursively calculated, we will use
Observation~\ref{obe:first} to directly obtain $f_m$:
$$f_m=(f_1+(m-1)(a-r))\mod a$$
Then we can use Lemma~\ref{lem:formular} to
directly calculate $d$, which is the answer to the original instance
of Problem~\ref{prob:nexthit}:
$$
d=(mp-f_0+f_m)/a
$$

However, the recursion (Equations~\ref{eqn:reduction}) may not always
be effective because the new progression's common distance $a-r$ may
not be much smaller than the old progression's common distance $a$ (We
will explain later why the size of the common distance is relevant to both
the time and space complexities of the {\tt NextHit}
algorithm). Fortunately, by using Lemma~\ref{lem:convert}, we can
overcome this difficulty by converting the recursive instance to
another instance which returns the same answer but works on a
progression with a smaller common distance. We summarize the two
possible reductions in the following.

\emph{Case 1}: $a-r\leq a/2$.  We want to work with
  $a-r$. Problem~\ref{prob:nexthit} is recursively reduced to a new
  instance of Problem~\ref{prob:nexthit} of a smaller size that finds
  $m$ over sequence $\bar{F}$ by setting:
$$p_{new}=a,\ \ \  a_{new}=a-r,\ \ \  u_{new}=f_1,\ \ \ L_{new}=L$$

\emph{Case 2: $r<a/2$.} We want to work with $r$. We first recursively
reduce Problem~\ref{prob:nexthit} to the same setting as in Case 1,
which will be further converted to the following parameter setting
because of Lemma~\ref{lem:convert}:

\hspace*{-6mm}
\begin{tabular}{c}
  \tallstrut$p_{new}=a,\ \ \  a_{new}=a-r,\ \ \  u_{new}=f_1,\ \ \ L_{new}=L$\\
  \tallstrut$\Downarrow \textit{(Lemma~\ref{lem:convert})}$\\
  \tallstrut$p_{new}=a, \  a_{new}=r, \ u_{new}=(a-f_1+L)\mod a, \ L_{new}=L$
\end{tabular}

\bigskip 

After adding some trivial recursion exit conditions, we present the
pseudocode of {\tt NextHit} in
Algo.~\ref{algo:nexthit}, which directly reflects the algorithmic
idea that we have presented.
Once {\tt NextHit} is clear, the calculation of  {\tt DirectSample}
becomes trivial and is presented in 
 Algo.~\ref{algo:directsample}.

\begin{algorithm}[t]
\small
\caption{$\ds(x,\ell,p,a,b)$}
\label{algo:directsample}
\KwIn{$x\geq 1$, $0\leq \ell \leq \maxlevel$, 
$p>0$, $0< a < p$, $0\leq b < p$ \tcp*[f]{$h(x)=(ax+b)\mod p$}}
\KwOut{$x+ \min\{i \geq 0 \mid h(x+i) \in R_\ell$, if $\min\{i \geq 0
  \mid h(x+i) \in R_\ell \}\neq \emptyset$; $-1$, otherwise.}

\BlankLine

$\mathcal{N} \leftarrow \nexthit(p,a,(ax+b)\mod p,
\lfloor 2^{-\ell}p \rfloor-1)$\;
\lIf{$\mathcal{N}=-1$}{\Return{$-1$}} \tcp*{This will not happen if
  $p$ is a prime number.}
\lElse{
  \Return{$x+\mathcal{N}$}\;
}
\end{algorithm}


\begin{algorithm}[t]
\small
\caption{$\nexthit(p, a, u, L)$}
\label{algo:nexthit}
\KwIn{$p>0$, $0\leq a < p$, $0\leq u < p$, $L\geq 0$}
\KwOut{$d = \min\{i\st 0\leq i\leq n, (u+i\cdot a)\text{ mod }p \leq
L\}$, if such $d$ exists; $-1$, otherwise.}

\BlankLine

\tcc*[f]{Recursive call exit conditions}

\lIf{$u\leq L$}{\Return{$d\leftarrow 0$}}\;
\label{line:exit1}
\lElseIf{$a=1$}{\Return{$d\leftarrow p-u$}}\;
\label{line:exit2}
\lElseIf{$a=0$}{\Return{$d\leftarrow -1$}}\;
\label{line:exit3}

\BlankLine

\tcc*[f]{Prepare for the recursive call: compute $|S_0|$, $f_1$, and $r$}

\lIf{$(p-u)\mod a = 0$}{$|S_0| \leftarrow (p-u)/a$}\;
\label{line:prep1}
\lElse{$|S_0| \leftarrow \lfloor (p-u)/a\rfloor +1$}\;
\label{line:prep2}
$f_1 \leftarrow (u+|S_0|\cdot a) \mod p$\;
\label{line:prep3}
$r \leftarrow p\mod a$\;
\label{line:prep4}

\BlankLine

\tcc*[f]{Recursive calls}

\lIf{$a-r \leq a/2$}{$d\leftarrow \nexthit(a,a-r,f_1,L)$}\tcp*{Case 1}
\label{line:recursion1}
\lElse {$d\leftarrow \nexthit(a,r,(a-f_1+R)\mod a,L)$}\tcp*{Case 2}
\label{line:recursion2}

\BlankLine

\tcc*[f]{Calculate and return $d$}

\lIf{$d=-1$}{\Return{$d$}\;}
\label{line:return1}
\Else{
  $f_{d+1} \leftarrow (f_1 + (a-r)*d)\mod a$\;
  \label{line:return2}
  \Return{
    $d \leftarrow (d*p+f_{d+1}+p-u)/a$\;
  }
  \label{line:return3}
}
\end{algorithm}

\begin{theorem}[Correctness and time and space complexity of
  $\nexthit$]
\label{thm:nexthit}
$\nexthit(p,a,u,L)$ solves Problem~\ref{prob:nexthit} using $O(\log
a)$ time and $O(\log p\cdot \log a)$ bits of space.  
When $p$ is a prime number, $\nexthit(p,a,u,L)$ solves
Problem~\ref{prob:nexthit} using $O(\min(\log a, \log d))$ time and
$O(\log p \cdot \min(\log a, \log d))$ bits of space, where $d$ is
the value returned by $\nexthit$.
\end{theorem}

\begin{proof}

  \emph{Correctness.}  Recall that $\nexthit(p, a, u, L)$ should
  return $d=\min\{i \geq 0 \mid (u+i\cdot a)\mod p \leq L\}$, if such
  $d$ exists; otherwise, it will return $d=-1$. The three exit conditions
  (Line~\ref{line:exit1}--~\ref{line:exit3}) capture all the possible
  cases where the algorithm can return and exit directly.  At the end of
  Line~\ref{line:exit3}, we know $S[d]$ does not occur in $S_0$ and
  $p>a\geq2$, so we are ready to reduce the task of finding $d$ to the
  task of finding $m$ over the sequence $\bar{F}$.
  Line~\ref{line:prep1}--~\ref{line:prep4} correctly calculate $f_1$
  and $r$ as the preparation work for the subsequent recursive calls
  at Line~\ref{line:recursion1}--~\ref{line:recursion2}.
  Line~\ref{line:recursion1} calculates the 0-based index of the first
  element in the sequence $\bar{F}$ that is not larger than $L$. By
  Lemma~\ref{lem:convert}, we know the recursive call at
  Line~\ref{line:recursion2} will return the same result as that from
  Line~\ref{line:recursion1}. In the case where $-1$ is returned by
  the recursive call (Line~\ref{line:recursion1}
  or~\ref{line:recursion2}), it means no element in the $\bar{F}$
  sequence is less than or equal to $L$. In that case, we will return
  $d=-1$ (Line~\ref{line:return1}).  Otherwise,
  Line~\ref{line:return2}--~\ref{line:return3} calculates and
  returns the value of $d$ using the results returned by the preceding recursive
  call (Lemma~\ref{lem:formular}).

  \emph{Time Complexity.}  We assume that the additions,
  multiplications, and divisions take unit time.  At each recursive
  call (Line \ref{line:recursion1} or \ref{line:recursion2}), we have
  $a_{new} \leq a/2$, plus the fact that in the worst case the
  recursion will return when $a=1$ (Line~\ref{line:exit2}), so the
  depth of the recursions is no more than $\log a$. 
 Because the time cost for the local computation
  in each recursive call is constant
  (Line~\ref{line:exit1}--\ref{line:prep4} and
  \ref{line:return1}--\ref{line:return3}), the time complexity of {\tt
    NextHit} is $O(\log a)$. 

  In the particular case where $p$ is a prime number, which is the
  case in the use of {\tt NextHit} for distributed basic counting in
  this paper, we know the value of $d$ returned by {\tt NextHit} is
  always non-negative (Lemma~\ref{lem:exist}). In this case, the
  length of the arithmetic progression that the caller of {\tt NextHit}
  works with is $d+1$. Because at each recursive call, we have the
  common distance in the new progression reduced by at least half, so
  the length of the progression that the next recursive {\tt NextHit} will work
  with is no more than a half of the caller's
  progression. This observation implies that the depth of the recursion
  is no more than $\log d$, so the overall time cost is bounded by
  $O(\log d)$. Comparing with the time cost of {\tt NextHit} in the
  general case, we see the time cost of {\tt NextHit} in this
  particular case is $O(\min(\log a, \log d))$.

  \emph{Space Complexity.}  In each recursive call, {\tt NextHit} needs
  to store a constant number of local variables such as $p, a, u, L$,
  etc. Since $p$ dominates $a$, $u$ and $L$ (if $L\geq p$, then
  $\nexthit()$ returns without recursive calls
  (Line~\ref{line:exit1}).), each recursive call needs $O(\log
  p)$ stack space. Since the depth of the recursion is no
  more than $O(\log a)$, which we have explained in the time complexity
  analysis, the space cost of the {\tt NextHit} algorithm is upper bounded
  by $O(\log^2 p)$ bits.  

  In the  case where $p$ is a prime number, the depth of the
  recursion is bounded by $O(\log d)$, which we have explained the time
  complexity analysis, so the total space cost is no
  more than $O(\log p \cdot \log d)$. Comparing with the space cost of
  {\tt NextHit} in 
  the general case, we get the time cost of {\tt NextHit} in this
  particular case is $O(\log p \cdot \min(\log a, \log d))$.
\end{proof}

\begin{corollary}[Correctness and time and space complexity of $\ds$]
\label{cor:directsample}
$\ds(x,\ell, p, a, b)$ finds the next $\ell$-level sample location on
or after the $x$th stream element using $O(\min(\log a, \log d))$ time
and $O(\log n\cdot \min(\log a, \log d))$  bits of space, where $n$ is an upper
bound of the stream size and $x+d$ is the value returned by $\ds$.
\end{corollary}

\begin{proof}
  The time and space cost of $\ds$ is dominated by the $\nexthit$
  subroutine. By Theorem~\ref{thm:nexthit} and combining the fact that
  $p$ is a random prime number chosen from $[10n, 20n]$, we can get
  the claim proved.
 \end{proof}

\section{Boosting the distributed basic counting via direct
  sampling}
\label{sec:new}

Now we present the complete picture of our new method for basic
counting on the union of multiple streams, followed by its correctness
proof and the analysis of its time, space, and communication
complexities.

\subsection{Algorithm description}
The formal description of the algorithm is presented in
Algo.~\ref{algo:new}. Note that one instance of the algorithm does not
automatically yield an $(\epsilon,\delta)$-estimate of $U$, but
produces an estimate  within a factor of $\epsilon$ of $U$
with a constant probability. The standard technique to reduce the
failure probability from a constant to the user-input parameter
$\delta$ is to run in parallel $O(\log(1/\delta))$ independent
instances of the algorithm and return the median of the results from
those instances.

The overall structure of the algorithm is still coordinated adaptive
sampling, but uses the hash function $h$ and thus can use the direct
sampling technique as a subroutine for a faster sampling procedure.
We first randomly pick a pairwise independent hash
function $h$ (Line~\ref{line:hash}) as defined in
Section~\ref{subsec:hash}. The hash function $h$ is shared by both the
stream processors (Alice and Bob) and the referee for coordinated
sampling. Each processor starts with  an empty sample of
size $\alpha=\samplesize$ and the sample level $\ell=0$
(Line~\ref{line:init}). 

\paragraph{The use of direct sampling.} 
The algorithm does not check the hash value of every stream element.
After processing a particular stream element $a_i$
(Line~\ref{line:if}--~\ref{line:discard}), the algorithm calls
$\ds(i+1,\ell, p, a, b)$ (Line~\ref{line:next}) to find the next
stream element that needs to be checked at the current sample level
$\ell$. The algorithm will go to sleep  until that element arrives.
That is, all the stream elements that do not have the possibility to
be selected regardless of its value  will be ignored
(Line~\ref{line:sleep}). When the element whose  location index
was returned by {\tt DirectSample} arrives, it will be selected into
the sample if it is a 1-bit 
(Line~\ref{line:if}). Note that we do not need to check the element's hash value 
as it has been guaranteed to be in $R_\ell$ by {\tt DirectSample}.
After the element is processed
(Line~\ref{line:if}--~\ref{line:discard}), the function call at
Line~\ref{line:next} gives the next stream location, at which the
stream element will need to be checked. By using the direct sampling
technique, our method intuitively is much faster than GT's method, which 
we will show later through both theoretical analysis and empirical study. 

In the case where the sample becomes overflow after the insertion of a
new element (Line~\ref{line:while}), the sample level will be increased by one
(Line~\ref{line:leveladd}). All the stream elements that are currently
in the sample but whose hash values do not belong to $R_\ell$ will be
discard from the sample (Line~\ref{line:discard}).  The sample level
will keep increasing until the sample is not
overflowed~(Line~\ref{line:while}). There is a possibility that the
sample level can exceed $\maxlevel$. If that happens, the algorithm
fails, but we will  later show the probability of this event is very
low.
The procedure at the referee to answer a query for $U$
(Line~\ref{line:referee-start}--~\ref{line:referee-end}) is the same
as the one in Algo.~\ref{algo:gt}.

\begin{algorithm}[t]
\small
  \caption{Distributed basic counting using direct sampling.}
  \label{algo:new}

\KwIn{Two geographically distributed bit streams $A=\{a_1, a_2, \ldots,
  a_n\}$ and $B=\{b_1, b_2, \ldots, b_n\}$, processed in parallel
  by  Alice and Bob, respectively.} 

\KwOut{$U(A,B) = \sum_{i=1}^n (a_i \lor b_i)$, returned by the referee.}

\smallskip

{\bf Randomly pick a pairwise independent hash function:
  \label{line:hash}
} 
$h(x) = (ax+b) \mod p$,
where $p$ is a prime number randomly picked from $[10n,20n]$, and
$a,b$ are randomly picked from $\{0,1,\ldots, p-1\}$ and $a\neq 0$.

\smallskip

\underline{\bf Alice}:

\smallskip 

Initialize an empty sample $S$ of size $\alpha=\samplesize$;
$\ell\leftarrow 0$\;
\label{line:init}

$i\leftarrow 1$
\label{line:start}
\tcp*{$\ds(1,0,p,a,b)\equiv 1$, the first location 
  to check at sample level $0$.}

\SetKwFor{SleepUntil}{Sleep until}{\ }{endSleep until}

\SleepUntil{$a_i$ {\bf arrives}}{
 \label{line:sleep}
  
  \If{$a_i=1$\label{line:if}}
 {
   $S \leftarrow S \cup \{i\}$\label{line:insert}\;
  \While (\tcp*[f]{Sample is full.}){$|S| > \alpha$\label{line:while}}{
      $\ell = \ell + 1$ \label{line:leveladd}\;
      \lIf {$\ell > \maxlevel$}{Exit} \tcp*{Algorithm
        fails.} \label{line:failnew} 
      Discard every $x\in S$ such that $h(x) \notin R_\ell$;
      \label{line:discard}
    }
  }
  $i \leftarrow \ds(i+1,\ell, p, a, b)$\label{line:next} \tcp*{The next location that needs
    to be checked.}
}

\smallskip

\underline{\bf Bob}: 
(Symmetric to Alice)

\smallskip

\underline{\bf Referee}: \label{line:referee} 
\tcp*[f]{Upon receiving a query for $U$.}

\smallskip 
Receive $S_A,\ell_A$ from Alice and $S_B,\ell_B$ from Bob; 
\label{line:referee-start}

$\ell^*\leftarrow\max(\ell_A, \ell_B)$;
\label{line:ell}

\lIf{$\ell_A < \ell^*$}{
     Discard every $x\in S_A$ such that 
        $h(x) \notin R_{\ell^*}$\;
}
\lElseIf{$\ell_B < \ell^*$}{
     Discard every $x\in S_B$ such that 
        $h(x) \notin R_{\ell^*}$\;
}
\Return{$|S_A \cup S_B| / P_{\ell^*}$}; \label{line:referee-end} 
\end{algorithm}

\subsection{Correctness}
Note that our method still uses the coordinated adaptive sampling on the high
level, but uses {\tt DirectSample} for speedup, so the correctness proof
of our method follows a parallel structure of the proof by Gibbons and
Tirthapura~\cite{GT01}. We present the entire proof here,
because the hash function used in our method
is different and also for completeness.

The only difference between Algo.~\ref{algo:gt} and~\ref{algo:new}
lies in the behavior of the stream processors. The processors in
Algo.~\ref{algo:gt} process every incoming stream element, whereas the
processors Algo.~\ref{algo:new} ignore those stream elements that have
no possibility to be selected into the sample. Due to the correctness
of {\tt DirectSample} used by Algo.~\ref{algo:new} for faster
sampling, given the same data stream and using the same hash function,
a processor using Algo.~\ref{algo:gt} will end up with the same
collection of selected elements as  a processor
using Algo.~\ref{algo:new} after processing the stream.  Recall
that the behavior of the referee from both algorithms is
identical. So, given the same distributed streams and the same hash
function for coordinated sampling, the answer returned by the referee
of Algo.~\ref{algo:gt} will be the same as the one returned by the
referee of Algo.~\ref{algo:new}. Thus, we can prove the correctness of 
Algo.~\ref{algo:new} by showing the correctness proof of
Algo.~\ref{algo:gt}.

Let $1_A=\{1\leq i \leq n \mid a_i = 1\}$, $1_B = \{1\leq i \leq n
\mid b_i = 1\}$, and $1_U = \{1\leq i \leq n \mid a_i \lor b_i = 1\}$.
Because the sampling procedures at both the processors and the referee
are coordinated, the sample obtained by the referee at the query time
can be viewed as the one whose elements are directly sampled by the
referee from the stream $1_U$. The quality of the answer returned by
the referee is solely determined by this sample, so our correctness proof
will be focused on the analysis of this sample.  Recall that $U(A,B) =
\sum_{i=1}^n (a_i \lor b_i)$ and thus $|1_U|=U(A,B)$. We will use $U$
to represent $U(A,B)$ if the context is clear.  The following process
is hypothetical, visualized only to serve the correctness proof, and
does not actually happen at the referee. The referee maintains
$\maxlevel + 1$ samples of level $\ell=0,1,\ldots,
\maxlevel$.  For each $i \in 1_U$, the element $i$ is
selected into the sample of level $\ell$ if and only if $h(i)\in R_\ell$.

\begin{definition}
  For $\ell=0,1,\ldots, \maxlevel$ and each $i \in 1_U$,
  let the indicator random variable
$$
  X_{\ell, i} = \left \{
\begin{array}{ll}
1, & \textit{ if $h(i)\in R_\ell$} \\
0, & \textit{ otherwise}
\end{array}
\right.
$$
\end{definition}

\begin{definition}
  For $\ell=0,1,\ldots, \maxlevel$, let the random variable
$X_\ell = \sum_{i\in 1_U} X_{\ell,i}$
\end{definition}

\begin{definition}
\label{def:bad}
For $\ell=0,1,\ldots, \maxlevel$,  
we say: (1) 
the sample of level $\ell$ is \emph{bad}, if $|X_\ell/P_\ell-U|>\epsilon U$. 
(2) random event $B_\ell$ happens if level $\ell$ is bad.
\end{definition}

\begin{fact}
\label{fact:pairwise}
For a particular $\ell \in \{0,1,\ldots, \maxlevel\}$
and any $i,j\in 1_U$ and $i\neq j$,
the random variables $X_{\ell, i}$ and $X_{\ell,j}$  are pairwise
independent. 
\end{fact}

\begin{fact}
\label{fact:pr}
  For any $\ell \in \{0,1,\ldots, \maxlevel\}$
  and any $i\in 1_U$: $1/2^{\ell+1}\leq \Pr[X_{\ell, i} = 1] = P_\ell\leq 1/2^{\ell}$
\end{fact}

\begin{lemma}
\label{lem:exp}
For each $\ell = 0,1,\ldots, \maxlevel$:
$\e[X_\ell] = UP_\ell$ and  $U/2^{\ell+1}\leq \e[X_\ell]\leq U/2^\ell$ 
\end{lemma}

\begin{proof}
$
\e[X_\ell] = \e\left[\sum_{i\in 1_U} X_{\ell, i}\right]
    = \sum_{i\in 1_U} \e\left[X_{\ell, i}\right] 
   = \sum_{i\in 1_U} \Pr\left[X_{\ell,i}=1\right]
= \sum_{i\in 1_U} P_\ell = UP_\ell
$, where the second equality uses the linearity of expectation. 
Using Fact~\ref{fact:pr}, the lemma is proved. 
\end{proof}

\begin{lemma}
\label{lem:var}
For each $\ell = 0,1,\ldots, \maxlevel$:
  $\var[X_\ell]\leq \e[X_\ell]$
\end{lemma}

\begin{proof}
  \begin{eqnarray*}
  \var[X_\ell] &=& \var\left[\sum_{i\in 1_U} X_{\ell, i}\right]\\
  &=& \sum_{i\in 1_U}\var\left[ X_{\ell, i}\right] 
  \textit{\ \ \ \ \ \ \ (Fact~\ref{fact:pairwise})}\\
 &=& \sum_{i\in 1_U} \left(\e\left[X^2_{\ell,i}\right]
    -\e^2\left[X_{\ell,i}\right]\right)\\
  &=& \sum_{i\in 1_U} \left(\Pr[X_{\ell,i}=1] 
     - \left(\Pr[X_{\ell,i}=1]\right)^2\right)\\
&\leq& \sum_{i\in 1_U} \Pr[X_{\ell,i}=1] \\
&=& \e[X_\ell]
 \end{eqnarray*}
\end{proof}

\begin{lemma}
\label{lem:b}
  For each $\ell = 0,1,\ldots, \maxlevel$: 
  $\Pr[B_\ell] \leq 2^{\ell+1} / (\epsilon^2 U)$
\end{lemma}

\begin{proof}
  \begin{eqnarray*}
    \Pr[B_\ell] 
    &=& \Pr\left[\left|\frac{X_\ell}{P_\ell} 
         - U\right|> \epsilon U\right] 
    = \Pr\left[\left| X_\ell - UP_\ell \right| 
        >\epsilon UP_\ell\right]\\
    &=& \Pr\left[\left| X_\ell - \e[X_\ell] \right| 
        >\epsilon \e[X_\ell]\right] 
     \textit{\ \ \ \ \ \ \ (Lemma~\ref{lem:exp})}\\
    &<& \frac{\var[X_\ell]}{(\epsilon \e[X_\ell])^2} 
     \textit{\ \ \ \ \ \ \ (Chebyshev's Inequality)}\\
    &\leq& \frac{1}{\epsilon^2 \e[X_\ell]}
     \textit{\ \ \ \ \ \ \ (Lemma~\ref{lem:var})}\\
   &=& \frac{2^{\ell+1}}{\epsilon^2 U}
     \textit{\ \ \ \ \ \ \ (Lemma~\ref{lem:exp})}
 \end{eqnarray*}

\end{proof}

\begin{definition}
\label{def:omega}
Let $\omega$ be the lowest numbered level such that $\e[X_\omega]<\alpha/2$.
\end{definition}

If $U\leq \alpha$, Algo.~\ref{algo:new} will certainly return
the exact value of $U$, so we only consider the interesting case where
$U > \alpha = \samplesize$.

\begin{lemma}
Level $\omega$ exists and  $0 < \omega < \maxlevel$.
\end{lemma}

\begin{proof}
(1) $\omega > 0$ is because $\e[X_0] = U > \alpha$
but $\e[X_\omega] < \alpha/2$.
(2) We prove $\omega < \maxlevel$ by
showing that $\e\bigl[X_{\maxlevel-1}\bigr] < \alpha /2$.
Note that $p$ is prime number from $[10n, 20n]$, so 
$\maxlevel-1 > \log n$.
It follows that 
\begin{eqnarray*}
\e\left[ X_{\maxlevel-1} \right] 
&\leq& \frac{U}{2^{\maxlevel-1}}  
  \textit{\ \ \ \ \ \ \ (Lemma~\ref{lem:exp})}\\
&\leq& \frac{U}{2^{\log n}}
\leq \frac{n}{2^{\log n}}
= 1 < \frac{\alpha}{2} = \frac{\samplesize}{2} 
\end{eqnarray*}

\end{proof}

\begin{lemma}
  \label{lem:omega}
    $\sum_{\ell=0}^{\omega}\Pr[B_\ell] < \frac{4}{15}$.
\end{lemma}

\begin{proof}
By Definition~\ref{def:omega}, we
have $\e[X_{\omega-1}] \geq \alpha/2 \geq 30/\epsilon^2$. 
Combining the fact that $\e[X_{\omega-1}] \leq U/2^{\omega-1}$
(Lemma~\ref{lem:exp}), we have
$U/2^{\omega-1} \geq 30/\epsilon^2$, 
i.e.,
$2^{\omega-1}/(\epsilon^2U)\leq 1/30$.
It follows that,
 \begin{eqnarray*}
    \sum_{\ell=0}^{\omega}\Pr[B_\ell] 
     &<&  \sum_{\ell=0}^{\omega} \frac{2^{\ell+1}}{\epsilon^2 U}
     \textit{\ \ \ \ \ \ \ (Lemma~\ref{lem:b})} \\ 
     &=& \frac{2}{\epsilon^2 U}\left(2^{\omega+1}-1\right)
     < \frac{8\cdot 2^{\omega-1}}{\epsilon^2 U}\leq \frac{4}{15}
  \end{eqnarray*}
\end{proof}

Recall that $\ell^*$ is the sample level at which the referee answer
the query (Line~\ref{line:l}, Algo.~\ref{algo:gt} and
Line~\ref{line:ell} of~\ref{algo:new}).

\begin{definition}
\label{def:s}
For $\ell = 0, 1, \ldots, \maxlevel + 1$, 
we say the random event $S_\ell$ happens if $\ell = \ell^*$, i.e., the
referee uses the sample  level $\ell^* = \ell$ to answers the query. 
\end{definition}

\begin{lemma}
  \label{lem:s}
$\Pr[S_{\omega+1} \lor \ldots \lor S_{\maxlevel + 1}]
\leq 1/30
$
\end{lemma}

\begin{proof}
%
  If the random event $S_{\omega+1} \lor \ldots \lor S_{\maxlevel+ 1}$ happens, it means the algorithm cannot use a level
 numbered smaller than $\omega+1$ to answer the query. It follows that 
  $X_\omega > \alpha$. So, 
 \begin{eqnarray*}
 \Pr\left[S_{\omega+1} \lor \ldots \lor S_{\maxlevel + 1}\right]
&\leq& \Pr[X_\omega > \alpha] 
= \Pr\left[X_\omega - \e[X_\omega] > \alpha - \e[X_\omega]\right] \nonumber\\
&\leq& \Pr \left[X_\omega - \e[X_\omega] > \alpha - \frac{\alpha}{2}\right]
\textit{\ \ \ \ \ \ \ \ (Definition~\ref{def:omega})}\nonumber\\
&\leq& \frac{\var[X_\omega]}{(\alpha/2)^2}
\textit{\ \ \ \ \ \ \ \ (Chebyshev's Inequality)}\nonumber\\
&\leq& \frac{\e[X_\omega]}{\alpha^2/4}
\textit{\ \ \ \ \ \ \ \ (Lemma~\ref{lem:var})}\\
&\leq& \frac{2}{\alpha}
\textit{\ \ \ \ \ \ \ \ (Lemma~\ref{def:omega})}\\
&=& \frac{\epsilon^2}{30}
\leq \frac{1}{30}
 \end{eqnarray*}
\end{proof}

\begin{lemma}
\label{lem:eps}
  One instance of Algo.~\ref{algo:new} returns an $\epsilon$-estimate of $U$
  with a constant probability of at
  least $3/10$. 
\end{lemma}

\begin{proof}
  The algorithm can fail for two possibilities: (1) it
  stops at a sample level $\ell \leq \maxlevel$, but the
  level is \emph{bad}; or (2) it reaches the sample level $\maxlevel + 1$ and the algorithm just exits
(Line~\ref{line:failnew}, Algo.~\ref{algo:new}). So,

  \begin{eqnarray*}
  \Pr[\textit{failure}] 
&=& \Pr\left[(S_0 \land B_0) \lor (S_1 \land B_1) \lor \ldots \lor
(S_{\maxlevel} \land B_{\maxlevel}) \lor
S_{\maxlevel + 1}\right] \\
&\leq& \Pr\left[B_0 \lor \ldots \lor B_{\omega}\right] 
+ \Pr\left[S_{\omega+1} \lor \ldots \lor S_{\maxlevel + 1}\right]\\
&\leq& \sum_{\ell=0}^{\omega} \Pr[B_\ell]
+ \Pr\left[S_{\omega+1} \lor \ldots \lor S_{\maxlevel + 1}\right]\\
&\leq& \frac{4}{15} + \frac{1}{30} 
\textit{\ \ \ \ \ \ \ \ (Lemma~\ref{lem:omega} and~\ref{lem:s})}\\
&=& \frac{3}{10}
 \end{eqnarray*}
\end{proof}

Let $\beta = \lceil 24\ln(1/\delta)\rceil$. The next theorem shows
that we can further reduce the failure probability of Algo.~\ref{algo:new} to
the user input parameter $\delta$ by running $\beta$ independent
instances of Algo.~\ref{algo:new} and returning the median of the answers
of those $\beta$ instances as the estimate of $U$.

\newpage

\begin{theorem}[Correctness of Algo.~\ref{algo:new}]
\label{thm:correctnew}
  The median of the answers returned by $\beta=\lceil 24\ln(1/\delta)\rceil$
  independent instances (using different hash functions picked randomly and
  independently) of Algo.~\ref{algo:new} is an
  $(\epsilon,\delta)$-estimate of $U$.
\end{theorem}

\begin{proof}
  For $i=1,2,\ldots,
  \beta$, let the indicator random variable $Y_i=1$ if the $i$th instance does
  not return an $\epsilon$-estimate; $Y_i = 0$, otherwise.  Due to
  Lemma~\ref{lem:eps}, we already have $\Pr[Y_i=1]\leq 3/10$.  Let
  binomial random variable $Y=\sum_{i=1}^{\beta} Y_i$, then
$$
  \e[Y] = \e\left[\sum_{i=1}^{\beta} Y_i\right] 
= \sum_{i=1}^{\beta} \e[Y_i] 
= \sum_{i=1}^{\beta} \Pr[Y_i=1] \leq \frac{3}{10} \beta
$$
 If the median of the $\beta$ independent answers 
 is
  not an $\epsilon$-estimate, it means that more than $\beta/2$ answers
  are not $\epsilon$-estimate, i.e., $Y > \beta / 2$. We want to bound 
$\Pr[Y > \beta/2]$.

For proof purpose, we define another binomial random variable
$X=\sum_{i=1}^{\beta} X_i$, where each $\Pr[X_i=1]=3/10$, and thus 
$\e[X] = 3 \beta / 10$.

 \begin{eqnarray*}
    \Pr\left[X > \frac{\beta}{2}\right]
&=& \Pr\left[X > \left(1+\frac 2 3\right)\e[X]\right]
\leq \exp\left(\frac{-\e[X](2/3)^2}{3}\right)
\textit{\ \ \ \ \ \ \ \ (Chernoff Bound)}\\
&\leq& \exp\left(-\frac{3}{10} 24\left(\ln\frac 1\delta\right) \frac 4{27}\right)
\leq \delta^{48/45} \leq \delta
 \end{eqnarray*}

Note that $X$ and $Y$ are both binomial
 random variables of the form $X=B(\beta, p_1)$ and 
$Y=B(\beta, p_2)$, but $p_1 = 3/10 \geq p_2$, 
so it is obvious: 
$$
\Pr\left[Y > \frac \beta 2\right] 
\leq \Pr\left[X > \frac \beta 2\right] 
\leq \delta
$$
That is, the probability that median of the answers returned by the
$\beta$ independent instances of Algo.~\ref{algo:new} is not an
$\epsilon$-estimate is no more than $\delta$.
\end{proof}

\subsection{Time complexity}
\label{subsec:time}
\begin{theorem}[Worst-case total time cost]
\label{thm:worst-total}
  Algo.~\ref{algo:new} spends $\Theta(n\log(1/\delta))$
  time in the worst case for processing a bit
  stream of size $n$ for maintaining  an
  $(\epsilon,\delta)$-estimate of $U$.
\end{theorem}

\begin{proof}
  The total time cost of Algo.~\ref{algo:new} for processing a stream
  includes the time cost for (1) the {\tt DirectSample} function calls
  (Line~\ref{line:next}), (2) inserting the selected 1-bits into the
  sample at various levels (Line~\ref{line:insert}), and (3)
  increasing the sample level when the sample is full
  (Line~\ref{line:while}--\ref{line:discard}).


  Suppose $d_1, d_2, \ldots, d_t$, for some $t\geq 1$, is the sequence
  of integers returned by the {\tt DirectSample} function calls during
  the stream processing.  Note that:
  \begin{equation}
    \label{eq:d}
    t+\sum_{i=1}^t d_i\geq n \textrm{\ \ \ \  and \ \ \ \ } t+\sum_{i=1}^{t-1} d_i < n
  \end{equation}
 The
  total processing time is bounded by
\begin{equation}
\label{eq:time}
\sum_{i=1}^{t} O(\log d_t) + O(t) + O\left(\alpha \maxlevel
\right)
\end{equation}
where the three terms  capture the time cost of type (1), (2),
and (3), respectively. Following Equation~\ref{eq:time}, 
\begin{eqnarray*}
&&\sum_{i=1}^{t} O(\log d_i) + O(t)+ O\left(\alpha \maxlevel\right)\\
&=&\sum_{i=1}^{t-1} O(\log d_i)+ O(\log d_t) 
+ O(t)
+ O\left(\samplesize \maxlevel\right)\\
&=& \sum_{i=1}^{t-1} O(d_i) + O(\log n)+ O(n)+ O\left(1/(\epsilon^2)\log n\right)\\
&=& O(n) + O(n) + O\left(1/(\epsilon^2)\log n\right)
\textit{\ \ (Inequality~\ref{eq:d})}\\
&=& O(n)
\end{eqnarray*}
The above bound is also tight. For example, when the stream has all
0-bits, the sample will always be empty and thus will always be at
level 0, meaning every stream element will be checked, giving
a total time cost of at least $\Theta(n)$.  Therefore, the worst case
total time cost of one instance of Algo.~\ref{algo:new} is
$\Theta(n)$.  The claim in the theorem  follows due to the fact
that we need to run $\beta$ instances of Algo.~\ref{algo:new}
(Theorem~\ref{thm:correctnew}) in order to maintain an
$(\epsilon,\delta)$-estimate of $U$.
\end{proof}

\paragraph{Comment:}  (1) $\Theta(n\log(1/\delta))$ is GT's total time cost
in any case, but it rarely occurs with our method. (2) It is necessary to 
spend $\Theta(n\log(1/\delta))$  in the worst case, because in
the case all stream elements are 0-bits, we need to  check every
stream element if we want to estimate $U$ with a relative error guarantee.

\begin{theorem}[Worst-case  per-element time cost]
\label{thm:per}
Algo.~\ref{algo:new} spends $\Theta(\log(1/\delta) \log n)$
time in the worst case for processing one stream element in order to  maintain
an $(\epsilon,\delta)$-estimate of $U$. 
\end{theorem}

\begin{proof}
  The  time cost of Algo.~\ref{algo:new} for processing one
  stream element  in the worst case includes those for (1) the 
  function call for {\tt DirectSample} (Line~\ref{line:next}),
  (2) inserting the element into the sample if it is a 1-bit
  (Line~\ref{line:insert}), and (3) increasing the
  sample level when the sample is full
  (Line~\ref{line:while}--\ref{line:discard}).

  The time cost of type (1) is bounded by $O(\log n)$. The time cost
  of type (2) is bounded by $O(1)$. Next, let's look at the time cost of
  type (3).  We organize all the elements in the sample into multiple
  linked lists.  Each linked list is the collection of elements who
  share the same highest sample level through which the
  element will still remain in the sample. We call this level as the
  list's \emph{surviving} level.  Note that there can be at most
  $\maxlevel+1$ linked lists, because the algorithm maintains no more
  than $\maxlevel+1$ sample levels.  Every time we need to insert a new element
  into the sample, there will be two cases.

  Case 1: the sample is not full. We will first find which linked list
  the new element belongs to. This takes $O(\maxlevel)=O(\log n)$
  time. We will then insert the new element into that linked list which
  takes $O(1)$ time. So altogether, the time cost of type (3) is
  $O(\log n)$.

  Case 2: if the sample is full, we will discard the elements from the
  sample in a \emph{lazy} manner. We will first keep increasing the
  sample level until we get at least one linked list $L$ whose
  surviving level is smaller than the current sample level. This takes
  no more than $O(\maxlevel)=O(\log n)$ time, because the algorithm
  maintains no more than $\maxlevel+1$ levels.  We will then delete an
  element from $L$.  Then we insert the new element into the
  appropriate linked list that the new element belongs to. This again
  takes no more than $O(\maxlevel) = O(\log n)$ time, since we need to
  find that right linked list among at most $\maxlevel+1$ linked
  lists. So altogether, the time cost of type (3) in this case is
  $O(\log n)$.  This bound is also tight.  For example, suppose at
  some certain point of time there are $\Theta(\log n)$ linked list
  being maintained in the sample.  In that case, finding the
  appropriate linked list to insert the new element will take
  $\Theta(\log n)$ time in the worst case, meaning the time cost of
  type (3) will be indeed $\Theta(\log n)$.

  Add the time cost of type (1), (2), and (3), we get the worst-case
  time cost for processing one stream element by one instance of
  Algo.~\ref{algo:new} is $\Theta(\log n)$.  The claim in the theorem
  then follows due to the fact that we need to run $\beta$ instances
  of Algo.~\ref{algo:new} (Theorem~\ref{thm:correctnew}) in order to
  maintain an $(\epsilon,\delta)$-estimate of $U$.
\end{proof}

  \paragraph{Comment:}  In practice, when the stream size is large, only a very small
  portion of the stream will be checked in our method.  This indicates the
  probability of any particular element being processed by more than a constant number of 
  instances of Algo.~\ref{algo:new} is very low. So, our method's worst-case
  time cost per element will be $O(\log n)$ with high probability.

\smallskip

Intuitively, the new ingredient {\tt DirectSample} can significantly
speed up the stream processing in practice. Next, we show that the
average total time cost of Algo.~\ref{algo:new} for processing a
$\gamma$-random bit stream is indeed sublinear of the stream size.
Recall that the notion of $\gamma$-random bit stream
(Definition~\ref{def:randstream}) can fit into the distribution of many
real-world streams.

\begin{theorem}[Average-case total time cost for $\gamma$-random  bit stream]
\label{thm:time-avg}
The average time cost of Algo.~\ref{algo:new} for processing 
a $\gamma$-random bit stream of size $n$
is   $O\bigl(\frac{1}{\gamma\epsilon^2}\log^2 n \log\frac {1}{\delta}\bigr)$ 
in the  maintenance of   an $(\epsilon,\delta)$-estimate for $U$.
%
\end{theorem}

\begin{proof}
  There are no more than $\maxlevel$ sample levels and the sample size
  is $\alpha=\samplesize$, so the number of 1-bits that are selected
  by one instance of Algo.~\ref{algo:new} over the course of the
  stream processing is no more than $\alpha \maxlevel$, which
 is also true even if the sample level exceeds $\maxlevel$ as the
  algorithm will exit in that case (Line~\ref{line:failnew}).

  On the other hand, in a $\gamma$-random bit stream, the probability
  that a  stream location selected by {\tt DirectSample} has a 1-bit is at least
  $\gamma$.  So, for each 1-bit that is selected over the
  course of stream processing, 
  we have in average 
$1/\gamma$\footnote{The expectation of a
    geometric random variable with parameter $\gamma$.} before {\tt
    DirectSample} instances of 
{\tt DirectSample} function calls
 before {\tt
    DirectSample} actually returns a stream location which indeed
  contains a 1-bit. Therefore, in average, there are no more than
  $(1/\gamma)\alpha \maxlevel$ stream elements that have been 
  checked by Algo.~\ref{algo:new}. By Theorem~\ref{thm:per}, we know the worst-case
  per-element time cost  by one
  instance of Algo.~\ref{algo:new} is $\Theta(\log n)$, we get the
  average-case total time cost by one instance of Algo.~\ref{algo:new}
  for processing a $\gamma$-random bit stream of size $n$ is 
   $O(\log n)\cdot (1/\gamma)\alpha \maxlevel =
  O((1/(\gamma\epsilon^2))\log^2 n)$.

  The claim in the theorem then follows from the fact
  that we need to run $\beta$ instances of Algo.~\ref{algo:new}
  (Theorem~\ref{thm:correctnew}) in order to guarantee an
  $(\epsilon,\delta)$-estimate guarantee of $U$.
\end{proof}

\paragraph{Comment:}  (1) Our method is significantly faster than GT's in
practice, providing the first sublinear-time algorithm in the average
case for the distributed basic counting problem over most real-world
streams. (2) Our method provides the users with the new tradeoff that
a lower accuracy demand (a larger value for $\epsilon$) promises a
faster processing speed, whereas GT's is
$\Theta(n\log(1/\delta))$ in any case.

\subsection{Space and communication complexity}

\begin{theorem}[Space and communication cost]
\label{thm:space}
  To maintain an $(\epsilon,\delta)$-estimate of $U$: (1) the workspace
  at each processor is $O((1/\epsilon^2 + \log n) \log(1/\delta)
  \log n)$ bits. (2) the memory usage by the referee is
  $O((1/\epsilon^2)\log(1/\delta)\log n)$ bits.  (3) The communication
  cost per query between each processor and the referee is
  $O((1/\epsilon^2)\log(1/\delta)\log n)$ bits.
\end{theorem}

\begin{proof} 
\emph{The workspace at processors.} 
The workspace by each processor is
$O((1/\epsilon^2 + \log n) \log(1/\delta) \log n)$ bits, including the
following memory usages. 

\begin{itemize}
\item
The memory space for storing the sample system. 
 The sample maintained by each processor contains no more than
  $\alpha=\samplesize$ integers from the range $[1,n]$, using no more
  than $O(\alpha \log n) = O((1/\epsilon^2) \log n)$ bits.   The memory
  usage for recording the hashing function is $O(\log n)$ bits,
  because the three parameters $a$, $b$, $p$ that define the hash
  function $h(x) = (ax+b) \mod p$ are all bounded $O(n)$ bits
  (Section~\ref{subsec:hash}).  The number of bits used for recording
  the current sample level is bounded by $O(\log \log p) = O(\log \log
  n)$, because the sample level is no more than $\maxlevel$. 
  There are $\beta = \lceil 24 \ln(1/\delta)\rceil$ instances running
  in parallel on each processor, so the total memory cost for the
  sample system at each processor is $O((1/\epsilon^2)\log n + \log n
  + \log n \log n)\log(1/\delta)) = O((1/\epsilon^2)\log(1/\delta)\log n)$ bits.

\item The stack memory cost for the {\tt DirectSample} function call
(Line~\ref{line:next},
  Algo.~\ref{algo:new}).  It has been proved in
  Corollary~\ref{cor:directsample} that the space usage by one {\tt
    DirectSample} function call is no more than $O(\log n\cdot
  \min(\log a, \log d)) = O(\log^2 n)$ bits.  There are $\beta =
  \lceil 24 \ln(1/\delta)\rceil$ instances running in parallel on each
  processor, so the total memory cost for the {\tt DirectSample}
  function calls is no more than $O(\log(1/\delta) \log^2 n)$ bits.
\end{itemize}

  \emph{The workspace at the referee.} The memory needs by the referee
  is only for storing the samples it has received from the two
  processors, where the space usage of the samples from each processor
  is $O((1/\epsilon^2)\log(1/\delta)\log n)$ bits.

\emph{The communication cost per query and per link.} Upon the arrival
of a query at the referee, the only information that the referee needs
to collect from the processors is the samples maintained by the
processors. So, the communication cost per query and per link is
$O((1/\epsilon^2)\log(1/\delta)\log n)$ bits, the space usage of the
samples sent from one processor.
\end{proof}

\paragraph{Comment:}  Compared with the space usage of GT's, each
processor's extra space cost in our method is $O(\log^2 n
\log(1/\delta))$ bits due to the ${\tt DirectSample}$ function calls.
This extra space usage is a lower order term compared with GT's space usage when 
the value of $\epsilon$ is reasonably small (for ex., $\epsilon\leq 0.1$).
Further, this extra space cost occurs only $O(\log n)$ times for each
algorithm instance in average during the stream processing,
 because there are only $O(\log n)$
{\tt DirectSample} function calls within one algorithm instance in
average over the course of stream processing.  In practice  the
stream size $n$ is often very large, which is usually the case in stream processing
(otherwise there is no need to design space-efficient streaming
algorithms), the stream locations that are returned by {\tt
  DirectSample} within one algorithm instance will be very sparse,
meaning the probability that those $\beta$ parallel instances will
call the {\tt DirectSample} function at the same time
is very low. So, in practice, the extra space used for stack maintenance
in the $\beta$ instances will be only $O(\log^2 n)$
bits and occurs only $O(\log n)$ times with high probability. 
We will later show through experiments with both real-world and
synthetic data that this  extra space usage is too negligible to be detected by the OS.

\subsection{Extension: multiple streams and multiple processors}
\label{sec:ext}
It is trivial to extend our method for maintaining an
$(\epsilon,\delta)$-estimate for the distributed basic counting over
$k>2$ streams that are  processed by $k$ processors, for some
constant $k$. The procedure at each of these $k$ processors will be
exactly the same as the procedure in the 2-stream scenario. The
difference is when a query arrives, the referee needs to collect the
samples and their sample levels from all the processors, and then run
the following procedure to generate the estimate for $U$
(Equation~\ref{eqn:uk}).  The correctness proof and time and space complexity analysis
for the 2-stream case can be directly applied to this $k$-stream
scenario.

{
\bigskip 
\noindent
\underline{\bf Referee:} \ \ \ \ \ {\tt // Upon receiving a query for U.}

\smallskip
\noindent
{\bf for} $i=1,2,\ldots, k$\\
\hspace*{5mm}Receive $(S_i$, $\ell_i)$ from Processor $i$;\\
$\ell^* \leftarrow \max\{\ell_i \mid 1\leq i \leq k\}$;\\
{\bf for} $i=1,2,\ldots, k$\\
\hspace*{5mm}{\bf if}($\ell_i < \ell^*$) {\bf then}
Discard every $x\in S_i$ such that $h(x)\notin R_{\ell^*}$;\\
{\bf return} $|S_1 \cup \ldots \cup S_k|/P_{\ell^*}$;
}

\section{Experiments}
\label{sec:exp}

\paragraph{Main message.} 
Note that our method does not change the accuracy of GT's coordinated
adaptive sampling technique (a.k.a.\ distinct sampling), whose
accuracy in estimating the basic counting has been well studied and
validated by prior work~\cite{Gi01vldb,CTX-SICOMP09}\footnote{In
  particularly, the hash function used in~\cite{CTX-SICOMP09} is the
  same as the one we use in this paper.}. Thus, in this section, we only
want to demonstrate the time and space efficiency of our method
compared with GT's through experiments with both real-world and
synthetic data. The main messages from the experimental study are: (1)
Our method is exponentially faster than GT's method, whose
processing time is linear of the stream size. (2) Our method's
processing speed increases when the accuracy demand decreases (the
value of $\epsilon$ increases), while GT's method's processing speed
does not change regardless of the value of $\epsilon$. 
(3) Our 
method does not introduce any detectable space overhead, regardless of
the data sets and the value of $\epsilon$, compared with the space
cost by GT's.
All the above observations are perfectly consistent with the
theoretical results
summarized in
Table~\ref{tab:compare}.

\paragraph{System setup and data source.} All of our experiments were
conducted on a Dell Precision T5500s machine that has a 2.13GHz
Quad-core E5506 CPU with 4M cache, but no parallelism was used. The
machine runs 64-bit Ubuntu 10.04 Desktop and has
8GB DDR3 ECC 1066MHz SDRAM. We faithfully implemented the coordinated
adaptive sampling both with and without using the {\tt DirectSample}
technique (Algo.~\ref{algo:gt} and~\ref{algo:new}) using the {\tt C++}
programming language\footnote{The {\tt C++} source code can be
  downloaded at:
  \url{http://penguin.ewu.edu/~bojianxu/publications}.}. All
executables were built by {\tt GCC 4.4.3}.  We used the following
real-world and synthetic bit sequences in our experiments:

\begin{itemize}
\item Audio
  Bible\footnote{\url{http://spiritlessons.com/Documents/Bible/NIV_MP3_Bible/NIV_MP3_Bible.zip}}.
  We concatenated all the MP3 files of the audio Bible. It has a total
  of $919,658,056$ bits, of which $460,805,446$ are 1-bits.

\item Video of President G.\ W.\ Bush's
  speech\footnote{\url{https://ia600306.us.archive.org/23/items/Political_videos-GeorgeWBush20030904_5_224/Political_videos-GeorgeWBush20030904_5_224.ogv}}.
  It has a total of $1,434,146,160$ bits, of which $710,447,850$ are
  1-bits.

\item NASA Earth
  image\footnote{\url{http://eoimages.gsfc.nasa.gov/images/imagerecords/73000/73909/world.topo.bathy.200412.3x21600x10800.png}}.
  It has a total of $1,575,903,872$ bits, of which $789,808,848$ are
  1-bits.

\item Day 37 of Worldcup 98's network
  traffic\footnote{\url{http://ita.ee.lbl.gov/html/contrib/WorldCup.html}}. It
  has a total of $893,788,160$ bits, of which $257,380,419$ are
  1-bits.

\item Two synthetic bit sequences, where the probabilities of having a
  1-bit are $0.3$ and $0.4$, respectively. Each sequence has a total
  of $1,000,000,000$ bits, of which $300,023,303$ and $400,002,206$
  are 1-bits, respectively.
\end{itemize}

\subsection{Time efficiency}

\begin{table*}
\centering
{\footnotesize
\begin{tabular}{|l|l|r|r|r|r|r|r|}
\hline
{\bf Data} & {\bf Methods}  & ${\bf \epsilon=0.01}$ 
& ${\bf \epsilon=0.02}$ 
 & ${\bf \epsilon=0.05}$ & ${\bf \epsilon=0.1}$ &
  ${\bf \epsilon=0.2}$ & ${\bf \epsilon=0.5}$ \\
  \hline
  \hline
  & \tallstrut GT's &$18.7823$&$18.3439$&$18.2506$&$18.2531$&$18.2386$&$18.2312$\\
  \hhline{~-------}
  Audio Bible&  \tallstrut Ours &$\03.6273$&$\01.0784$&$\00.2841$&$\00.0528$&$\00.0285$&$\00.0046$\\
  \hhline{~-------}
  $460,805,446$ 1-bits &  \tallstrut {\bf Speedup} &${\bf >5}$x&${\bf >17}$x&${\bf >64}$x&${\bf >345}$x&${\bf >639}$x&${\bf >3952}$x\\
  \hline
  & \tallstrut GT's &$28.8833$&$28.6107$&$28.4632$&$28.2885$&$28.4486$&$28.5082$\\
  \hhline{~-------}
  Video of President&  \tallstrut Ours &$\03.2259$&$\01.0117$&$\00.3180$&$\00.1055$&$\00.0278$&$\00.0076$\\
  \hhline{~-------}
  $710,447,850$ 1-bits&  \tallstrut {\bf Speedup} &${\bf >8}$x&${\bf >28}$x&${\bf >89}$x&${\bf >268}$x&${\bf >1020}$x&${\bf >3727}$x\\
  \hline
 &  \tallstrut GT's &$31.9561$&$31.6179$&$31.5310$&$31.4197$&$31.5040$&$31.5009$\\
  \hhline{~-------}
  Earth Image &  \tallstrut Ours &$\03.1705$&$\01.3294$&$\00.2932$&$\00.0843$&$\00.0245$&$\00.0066$\\
  \hhline{~-------}
  $789,808,848$ 1-bits&  \tallstrut {\bf Speedup} &${\bf >10}$x&${\bf >23}$x&${\bf >107}$x&${\bf >372}$x&${\bf >1285}$x&${\bf >4725}$x\\
  \hline
  &  \tallstrut GT's &$12.9562$&$12.5322$&$12.5236$&$12.5224$&$12.5106$&$12.4987$\\
  \hhline{~-------}
  Worldcup 98&  \tallstrut Ours &$5.4275$&$\01.6266$&$\00.2791$&$\00.0678$&$\00.0372$&$\00.0070$\\
  \hhline{~-------}
  $257,380,419$ 1-bits&  \tallstrut {\bf Speedup} &${\bf >2}$x&${\bf >7}$x&${\bf >44}$x&${\bf >184}$x&${\bf >336}$x&${\bf >1762}$x\\
 \hline
&  \tallstrut GT's &$15.6868$&$15.2895$&$15.2082$&$15.1802$&$15.1939$&$15.1601$\\
  \hhline{~-------}
 \tallstrut Synthetic-0.3 &  \tallstrut Ours &$\02.9945$&$\01.4045$&$\00.4531$&$\00.0998$&$\00.0531$&$\00.0078$\\
  \hhline{~-------}
  $300,023,303$ 1-bits &  \tallstrut {\bf Speedup} &${\bf >5}$x&${\bf >10}$x&${\bf >33}$x&${\bf >152}$x&${\bf >286}$x&${\bf >1923}$x\\
  \hline
  &  \tallstrut GT's &$18.5577$&$17.5786$&$18.0129$&$17.9844$&$17.9933$&$17.9982$\\
  \hhline{~-------}
  \tallstrut Synthetic-0.4&  \tallstrut Ours &$\02.3465$&$\00.8511$&$\00.2100$&$\00.0782$&$\00.0306$&$\00.0063$\\
  \hhline{~-------}
$400,002,206$ 1-bits  &  \tallstrut {\bf Speedup} &${\bf >7}$x&${\bf >20}$x&${\bf >85}$x&${\bf >229}$x&${\bf >586}$x&${\bf >2851}$x\\
  \hline
\end{tabular}
} 
\caption{The time cost (in seconds) of both GT's and our methods in the
  processing of multiple bit streams with different accuracy demands
  $\epsilon$. The processing speed of GT's does not change with
  $\epsilon$, while ours becomes faster when $\epsilon$
  increases.  In all cases, our method  is significantly faster than
  GT's method.}
\label{tab:eps-time}
\end{table*}

\begin{figure*}
\centering
\hspace*{7mm}
  \subfigure[Audio Bible]{\label{fig:size-time-eps0.01-1} \includegraphics[scale=1.0]{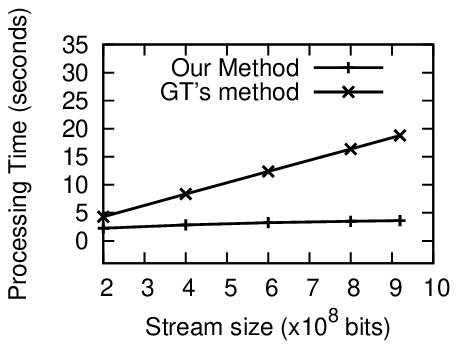}}%
  \subfigure[Video of President]{\label{fig:size-time-eps0.01-2} \includegraphics[scale=1.0]{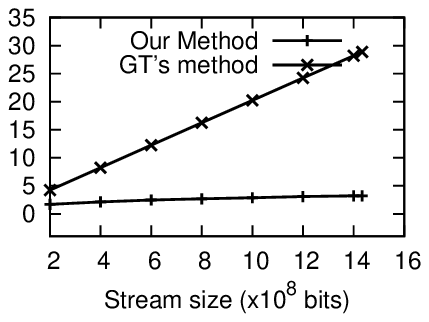}}%
  \subfigure[Earth Image]{\label{fig:size-time-eps0.01-3} \includegraphics[scale=1.0]{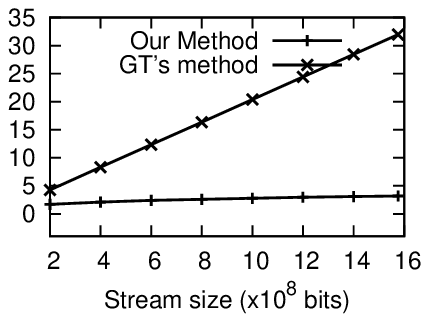}}%
  \newline
  \hspace*{5mm}
  \subfigure[Worldcup 98]{\label{fig:size-time-eps0.01-4} \includegraphics[scale=1.0]{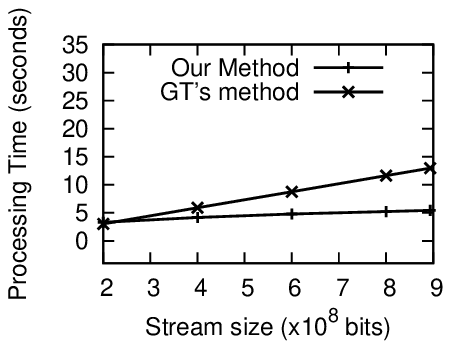}}%
  \subfigure[Synthetic-0.3]{\label{fig:size-time-eps0.01-5} \includegraphics[scale=1.0]{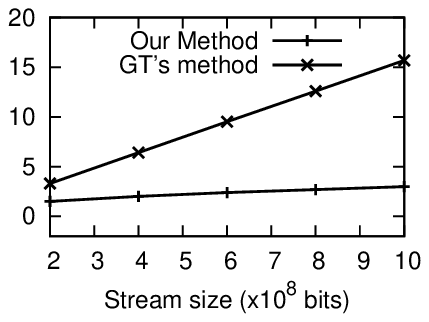}}%
  \subfigure[Synthetic-0.4]{\label{fig:size-time-eps0.01-6} \includegraphics[scale=1.0]{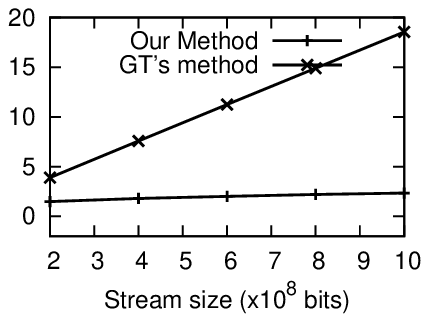}}%
  \caption{Stream size vs.\ time, $\epsilon = 0.01$. The processing
    time of GT's method is linear of the stream size, while ours is
    sublinear. Our method overall is much faster
    than GT's, especially when the stream size (the number of 1-bits
    in the stream, indeed) becomes larger.}
\label{fig:size-time-eps0.01}
\end{figure*}

Table~\ref{tab:eps-time} shows the total time cost of both methods in
processing different data sets with different accuracy demands.

\emph{Overall boosting.} Our method overall is significantly faster than GT's by a factor of several
  to several thousand times. This speedup becomes more significant when the value of
  $\epsilon$ increases and/or the stream size (more precisely, the
  number of 1-bits in the stream) increases. When deployed in the real
  world, where the stream size is nearly unbounded, our method can
  save a a vast majority of the processing time and energy that is
  needed by GT's.

\emph{A new tradeoff.} The time cost of GT's method is independent from the value of
  $\epsilon$, because their technique processes every stream element.
  However, our method becomes much faster when the
  value of $\epsilon$ increases. This
  is because a larger $\epsilon$ yields a smaller sample size, which makes the
  sample level be increased more often. Recall that a
  higher sample level selects stream elements with smaller
  probability, so it helps our {\tt DirectSample} technique be able to skip
  more stream elements.  This new feature in our method is
  important and useful, because it provides the user with 
  a new trade-off that a lower accuracy demand will not only save memory
  space but also will speed up the data processing.

\emph{Even faster for longer streams.} The more 1-bits are present in the stream, the more significant improvement  our
  method makes against GT's method. This is because more 1-bits
  leads to a higher sample level during the stream processing, which
  yields a lower sampling probability. That helps  {\tt
    DirectSample} skip over more stream elements,
  leading to an overall faster  processing speed.

\emph{Sublinear-linear time cost.} Figure~\ref{fig:size-time-eps0.01}
shows that the time cost of GT's method is linear of the stream size,
simply because their method processes every stream element, whereas
our method's processing time is sublinear of the stream size.  
The figure again shows our method overall is much faster than GT's, especially when the
stream size (the number of 1-bits in the stream, indeed) becomes
larger for the reason that we have given above.  
The plot in Figure~\ref{fig:size-time-eps0.01} shows the case where
$\epsilon=0.01$.  Figures regarding other $\epsilon$ values are given
in the appendix, from
which the same observations can be made.

\subsection{Space efficiency}

\begin{figure*}
\centering

\begin{minipage}[c]{1.4in}
\centering
\hspace*{-16mm}
\subfigure{\label{fig:eps-space-plot} \includegraphics[scale=1.1]{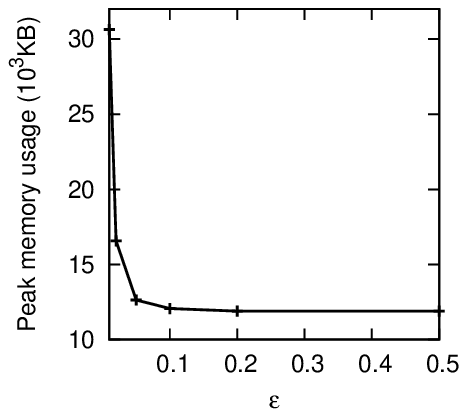}}
\end{minipage}
\hspace*{4mm}
\begin{minipage}[c]{4.8in}
\subfigure{\label{fig:eps-space-table} 
\centering
{\footnotesize
\begin{tabular}{|l|r|r|r|r|r|r|}
\hline
{\bf Methods}  & ${\bf \epsilon=0.01}$ 
& ${\bf \epsilon=0.02}$ 
 & ${\bf \epsilon=0.05}$ & ${\bf \epsilon=0.1}$ &
  ${\bf \epsilon=0.2}$ & ${\bf \epsilon=0.5}$ \\
  \hline
  \hline
  \tallstrut GT's & $30,640$ & $16,576$ & $12,640$ & $12,072$ & $11,892$ & $11,892$\\
  \hline
    \tallstrut Ours & $30,640$ & $16,576$ & $12,640$ & $12,072$ & $11,892$ & $11,892$\\
  \hline
    \tallstrut {\bf Overhead} & ${\bf 0\%}$ & ${\bf 0\%}$ & ${\bf 0\%}$ & ${\bf 0\%}$ & ${\bf 0\%}$ & ${\bf 0\%}$ \\
  \hline
\end{tabular}
}
}
\caption{The space usage (in KB) of both
    methods in the processing of multiple bit streams with  different
   $\epsilon$. The space usage of both methods
    increases when $\epsilon$ decreases and is independent from the
    stream size.  There is no detectable extra
    space cost by our method compared with GT's.}
\end{minipage}

\label{fig:eps-space}
\end{figure*}

We measure the space usage of our programs by {\tt VmPeak} minus the
memory cost for storing the data sets.  {\tt VmPeak} is
an entry in the {\tt /proc/<pid>/status} file, provided by the Linux
system. It captures the peak usage in KB of the total amount of virtual
memory used by the process, including the memory cost for the code,
data, and shared libraries plus the pages that have been swapped out.
{\tt VmPeak} represents the peak of the actual total memory cost of the process.
Figure~\ref{fig:eps-space} shows the space usage of both methods.

\emph{Space usage only depends on $\epsilon$.} The space usages of both methods are independent
  of stream size, but heavily depend upon the value of
  $\epsilon$. When the value of $\epsilon$ decreases, the space usage
  increases. This is consistent with the theoretical results --- the
  space usage of one instance of GT's and our method are
  $O((1/\epsilon^2)\log n)$ bits and $O((1/\epsilon^2 + \log n)\log
  n)$ bits, respectively. Note that we use the 64-bit {\tt unsigned
    long int} to represent the stream size, so the impact of the $\log
  n$ term in the big-oh bounds becomes fixed.

\emph{No detectable extra space usage.} For all data sets and all  $\epsilon$ values,
  there is no   detectable extra space usage by our method.
 That is, the extra space cost of
  $O(\log^2 n)$ bits by our method from the theoretical analysis is
  too negligible to be detected  by the OS in practice.

\section{Conclusion}
\label{sec:con}
In this paper, we designed {\tt DirectSample}, a new technique for
fast sampling, and used it in the coordinated adaptive sampling for
distributed basic counting.  Both the theoretical analysis and the
experimental results show that our method is exponentially faster than
the state-of-the-art GT's method. Further, the extra space usage by
our method from the theoretical analysis is not only negligible in
theory but also undetectable in practice.  Our new method can
potentially save a vast majority of processing time and energy needed
by GT's method in the processing of real-world streams, whose size is
nearly unbounded.  We also believe the new {\tt DirectSample}
technique can be of other independent interest.


\section{Acknowledgement}
The author acknowledges Atalay Mert \.{I}leri from
Bilkent University for the helpful discussion on the proof of
Lemma~\ref{lem:mapping}.

\bibliographystyle{plain}

\bibliography{streams}

\begin{thebibliography}{10}

\bibitem{AGMS99}
N.~Alon, P.~B. Gibbons, Y.~Matias, and M.~Szegedy.
\newblock Tracking algorithms for join and self-join sizes.
\newblock In {\em Proc.~18th ACM Symp.~on Principles of Database Systems},
  pages 1--11, May 1999.
\newblock Full version to appear in JCSS special issue for PODS'99.

\bibitem{AMS99}
N.~Alon, Y.~Matias, and M.~Szegedy.
\newblock The space complexity of approximating the frequency moments.
\newblock {\em Journal of Computer and System Sciences}, 58(1):137--147, 1999.

\bibitem{BBDMW02}
B.~Babcock, S.~Babu, M.~Datar, R.~Motwani, and J.~Widom.
\newblock Models and issues in data stream systems.
\newblock In {\em Proc.~21st ACM Symp.~on Principles of Database Systems
  (PODS)}, pages 1--16, 2002.

\bibitem{BJKST02}
Z.~Bar-Yossef, T.~Jayram, R.~Kumar, D.~Sivakumar, and L.~Trevisan.
\newblock Counting distinct elements in a data stream.
\newblock In {\em Proc.~6th International Workshop on Randomization and
  Approximation Techniques (RANDOM)}, pages 1--10, 2002.
\newblock Lecture Notes in Computer Science, vol.~2483, Springer.

\bibitem{BT07}
C.~Busch and S.~Tirthapura.
\newblock A deterministic algorithm for summarizing asynchronous streams over a
  sliding window.
\newblock In {\em Proc. 24th Annual Symposium on Theoretical Aspects of
  Computer Science (STACS)}, pages 465--476, 2007.

\bibitem{CW79}
J.L. Carter and M.L. Wegman.
\newblock Universal classes of hash functions.
\newblock {\em Journal of Computer and System Sciences}, 18(2):143--154, 1979.

\bibitem{CGMR05}
G.~Cormode, M.~Garofalakis, S.~Muthukrishnan, and R.~Rastogi.
\newblock Holistic aggregates in a networked world: Distributed tracking of
  approximate quantiles.
\newblock In {\em {ACM} {SIGMOD} International Conference on Management of Data
  ({SIGMOD})}, pages 25--36, 2005.

\bibitem{CMY08}
G.~Cormode, S.~Muthukrishnan, and K.~Yi.
\newblock Algorithms for distributed, functional monitoring.
\newblock In {\em {ACM}-{SIAM} Symposium on Discrete Algorithms ({SODA})},
  2008.

\bibitem{CMZ06}
G.~Cormode, S.~Muthukrishnan, and W.~Zhuang.
\newblock What's different: Distributed, continuous monitoring of
  duplicate-resilient aggregates on data streams.
\newblock In {\em International Conference on Data Engineering ({ICDE})}, pages
  20--31, 2006.

\bibitem{CTX-SICOMP09}
G.~Cormode, S.~Tirthapura, and B.~Xu.
\newblock Time-decaying sketches for robust aggregation of sensor data.
\newblock {\em SIAM Journal on Computing}, 39(4):1309--1339, 2009.
\newblock (Also in PODC2007).

\bibitem{DGIM02-SICOMP}
M.~Datar, A.~Gionis, P.~Indyk, and R.~Motwani.
\newblock Maintaining stream statistics over sliding windows.
\newblock {\em SIAM Journal on Computing}, 31(6):1794--1813, 2002.

\bibitem{Fischer01theart}
E.~Fischer.
\newblock The art of uninformed decisions: A primer to property testing.
\newblock {\em Science}, 75:97--126, 2001.

\bibitem{Gi01vldb}
P.~Gibbons.
\newblock Distinct sampling for highly-accurate answers to distinct values
  queries and event reports.
\newblock In {\em Proc.~27th International Conf.~on Very Large Data Bases
  (VLDB)}, pages 541--550, 2001.

\bibitem{GT01}
P.~Gibbons and S.~Tirthapura.
\newblock Estimating simple functions on the union of data streams.
\newblock In {\em Proc.~ACM Symp.~on Parallel Algorithms and Architectures
  (SPAA)}, pages 281--291, 2001.

\bibitem{GT02}
P.~Gibbons and S.~Tirthapura.
\newblock Distributed streams algorithms for sliding windows.
\newblock In {\em Proc. ACM Symposium on Parallel Algorithms and Architectures
  (SPAA)}, pages 63--72, 2002.

\bibitem{GKMS01}
A.~C. Gilbert, Y.~Kotidis, S.~Muthukrishnan, and M.~J. Strauss.
\newblock Surfing wavelets on streams: one-pass summaries for approximate
  aggregate queries.
\newblock In {\em Proc.~27th International Conf.~on Very Large Data Bases
  (VLDB)}, pages 79--88, 2001.

\bibitem{Goldreich98}
O.~Goldreich.
\newblock Combinatorial property testing (a survey).
\newblock In {\em In: Randomization Methods in Algorithm Design}, pages 45--60.
  American Mathematical Society, 1998.

\bibitem{IW03}
P.~Indyk and D.~Woodruff.
\newblock Tight lower bounds for the distinct elements problem.
\newblock In {\em Proc. 44th IEEE Symp. on Foundations of Computer Science
  (FOCS)}, page 283, 2003.

\bibitem{KN97}
E.~Kushilevitz and N.~Nisan.
\newblock {\em Communication Complexity}.
\newblock Cambridge University Press, Cambridge, UK, 1997.

\bibitem{MG82}
J.~Misra and D.~Gries.
\newblock Finding repeated elements.
\newblock {\em Science of Computer Programming}, 2:143--152, November 1982.

\bibitem{muthu-book}
S.~Muthukrishnan.
\newblock {\em Data Streams: Algorithms and Applications}.
\newblock Foundations and Trends in Theoretical Computer Science. Now
  Publishers, August 2005.

\bibitem{PT07}
A.~Pavan and S.~Tirthapura.
\newblock Range-efficient counting of distinct elements in a massive data
  stream.
\newblock {\em SIAM Journal on Computing}, 37(2):359--379, 2007.
\newblock (Also in ICDE2005).

\bibitem{DanaRon08-survey}
D.~Ron.
\newblock Property testing: A learning theory perspective.
\newblock {\em Foundations and Trends in Machine Learning}, 1(3):307--402,
  2008.

\bibitem{RS-SIDMA11}
R.~Rubinfeld and A.~Shapira.
\newblock Sublinear time algorithms.
\newblock {\em SIAM Journal on Discrete Mathematics}, 25(4):1562--1588, 2011.

\bibitem{SBAS04}
N.~Shrivastava, C.~Buragohain, D.~Agrawal, and S.~Suri.
\newblock Medians and beyond: new aggregation techniques for sensor networks.
\newblock In {\em SenSys}, pages 239--249, 2004.

\bibitem{XTB08}
B.~Xu, S.~Tirthapura, and C.~Busch.
\newblock Sketching asynchronous data streams over sliding windows.
\newblock {\em Distributed Computing}, 20(5):359--374, 2008.

\end{thebibliography}

\newpage
\appendix

\section*{Appendix}

Additional figures plotting the processing time of both methods
regarding the stream size over multiple data sets and
different values for $\epsilon$.  The time cost of GT's
method is linear of the stream size, simply because their method
processes every stream element, whereas our method's processing time
is sublinear of the stream size.  Also, our method overall is much
faster than GT's, especially when the stream size (the number of
1-bits in the stream, indeed) becomes larger. 

\bigskip 

\begin{figure*}[h!]
  \centering
  \subfigure[Audio Bible]{\label{fig:size-time-eps0.02-1} \includegraphics[scale=1.0]{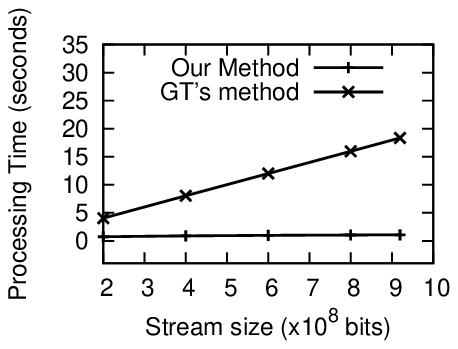}}%
  \subfigure[Video of President]{\label{fig:size-time-eps0.02-2} \includegraphics[scale=1.0]{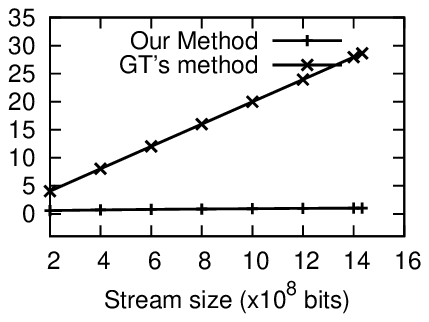}}%
  \subfigure[Earth Image]{\label{fig:size-time-eps0.02-3} \includegraphics[scale=1.0]{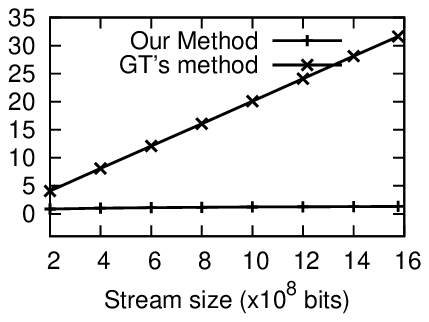}}%
  \newline
  \subfigure[Worldcup 98]{\label{fig:size-time-eps0.02-4} \includegraphics[scale=1.0]{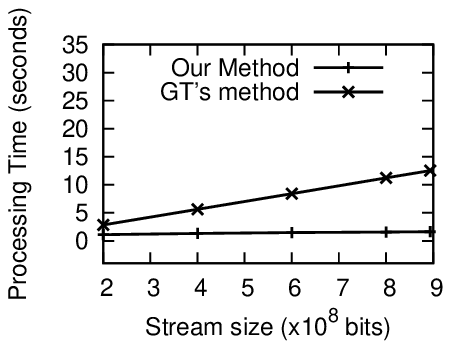}}%
  \subfigure[Synthetic-0.3]{\label{fig:size-time-eps0.02-5} \includegraphics[scale=1.0]{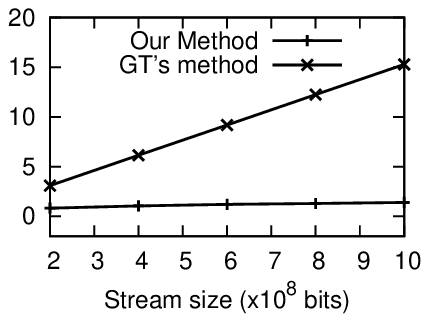}}%
  \subfigure[Synthetic-0.4]{\label{fig:size-time-eps0.02-6} \includegraphics[scale=1.0]{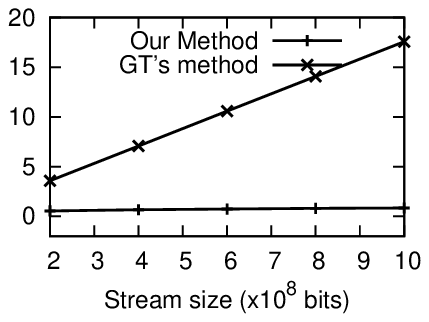}}%
\caption{Stream size vs.\ time, $\epsilon = 0.02$}
\label{fig:size-time-eps0.02}
\end{figure*}

\begin{figure*}[h!]
  \centering
  \subfigure[Audio Bible]{\label{fig:size-time-eps0.05-1} \includegraphics[scale=1.0]{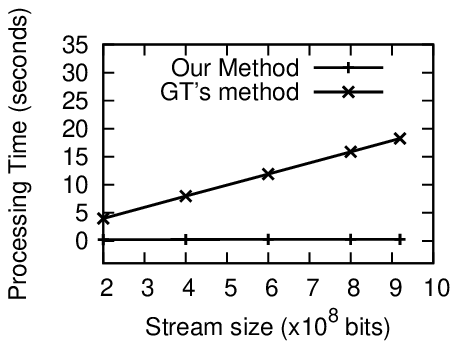}}%
  \subfigure[Video of President]{\label{fig:size-time-eps0.05-2} \includegraphics[scale=1.0]{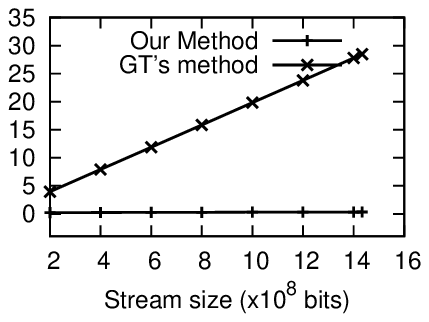}}%
  \subfigure[Earth Image]{\label{fig:size-time-eps0.05-3} \includegraphics[scale=1.0]{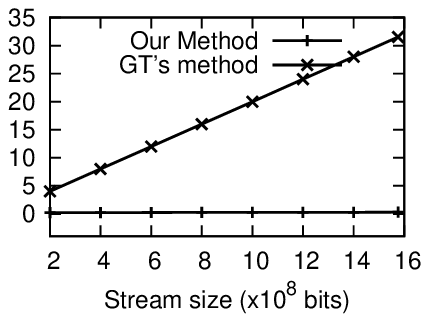}}%
  \newline
  \subfigure[Worldcup 98]{\label{fig:size-time-eps0.05-4} \includegraphics[scale=1.0]{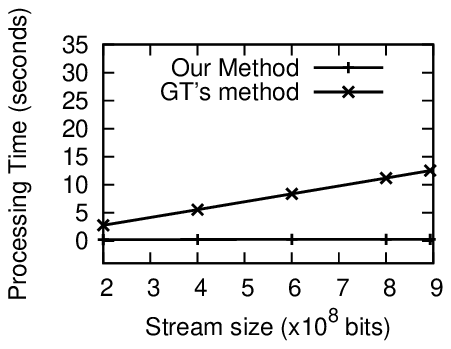}}%
  \subfigure[Synthetic-0.3]{\label{fig:size-time-eps0.05-5} \includegraphics[scale=1.0]{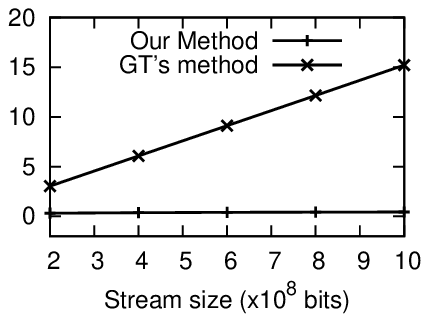}}%
  \subfigure[Synthetic-0.4]{\label{fig:size-time-eps0.05-6} \includegraphics[scale=1.0]{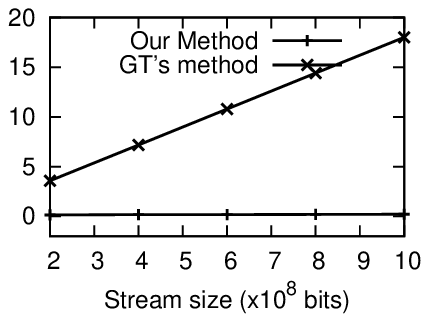}}%
\caption{Stream size vs.\ time, $\epsilon = 0.05$}
\label{fig:size-time-eps0.05}
\end{figure*}

\begin{figure*}[h!]
  \centering
  \subfigure[Audio Bible]{\label{fig:size-time-eps0.1-1} \includegraphics[scale=1.0]{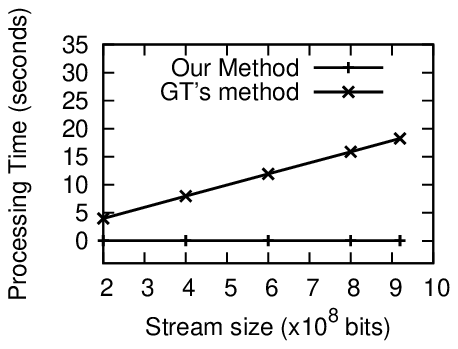}}%
  \subfigure[Video of President]{\label{fig:size-time-eps0.1-2} \includegraphics[scale=1.0]{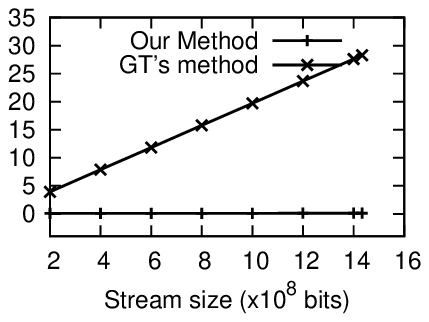}}%
  \subfigure[Earth Image]{\label{fig:size-time-eps0.1-3} \includegraphics[scale=1.0]{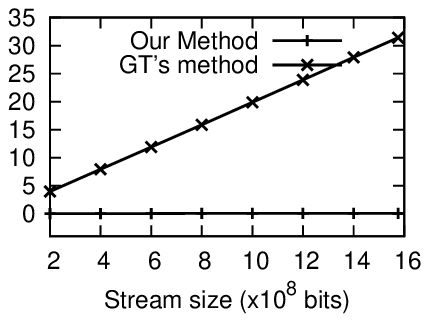}}%
  \newline
  \subfigure[Worldcup 98]{\label{fig:size-time-eps0.1-4} \includegraphics[scale=1.0]{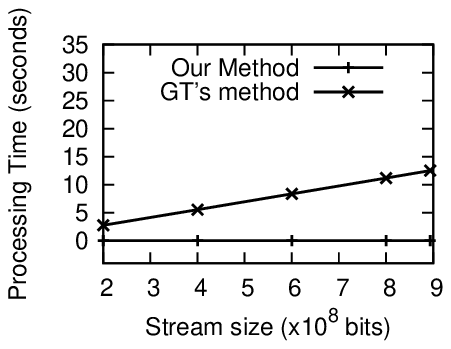}}%
  \subfigure[Synthetic-0.3]{\label{fig:size-time-eps0.1-5} \includegraphics[scale=1.0]{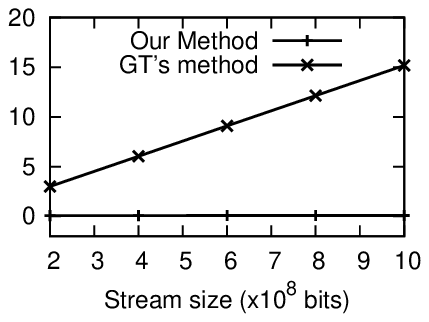}}%
  \subfigure[Synthetic-0.4]{\label{fig:size-time-eps0.1-6} \includegraphics[scale=1.0]{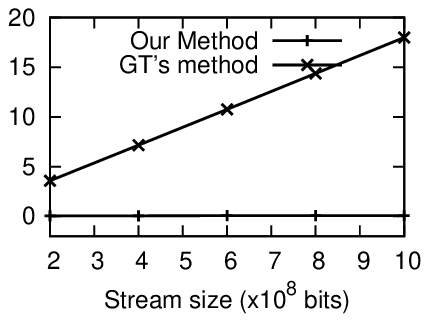}}%
\caption{Stream size vs.\ time, $\epsilon = 0.1$}
\label{fig:size-time-eps0.1}
\end{figure*}

\begin{figure*}[h!]
  \centering
  \subfigure[Audio Bible]{\label{fig:size-time-eps0.2-1} \includegraphics[scale=1.0]{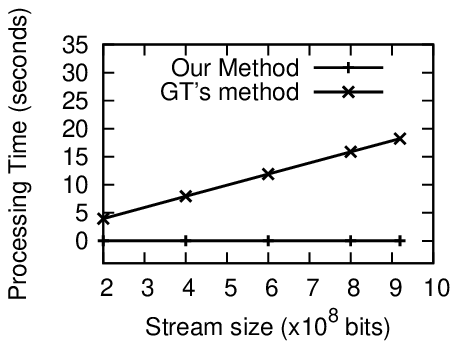}}%
  \subfigure[Video of President]{\label{fig:size-time-eps0.2-2} \includegraphics[scale=1.0]{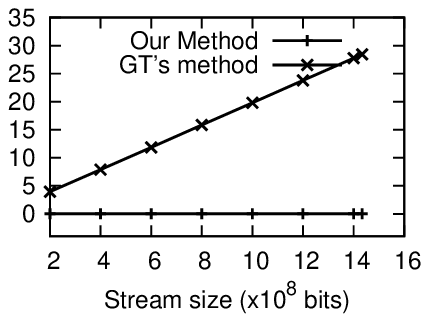}}%
  \subfigure[Earth Image]{\label{fig:size-time-eps0.2-3} \includegraphics[scale=1.0]{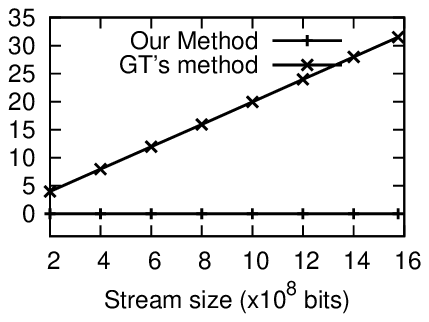}}%
  \newline
  \subfigure[Worldcup 98]{\label{fig:size-time-eps0.2-4} \includegraphics[scale=1.0]{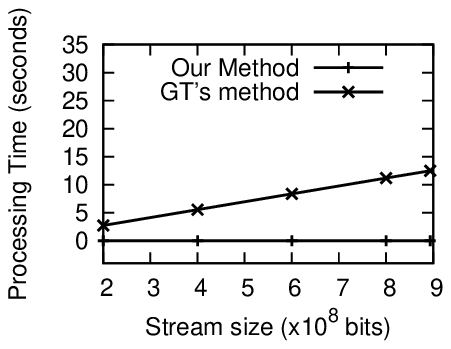}}%
  \subfigure[Synthetic-0.3]{\label{fig:size-time-eps0.2-5} \includegraphics[scale=1.0]{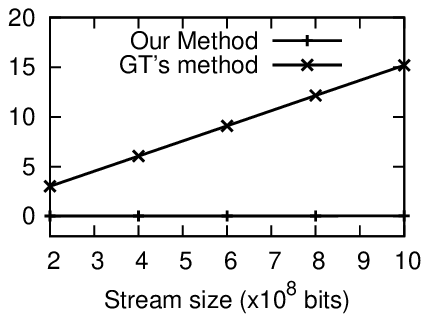}}%
  \subfigure[Synthetic-0.4]{\label{fig:size-time-eps0.2-6} \includegraphics[scale=1.0]{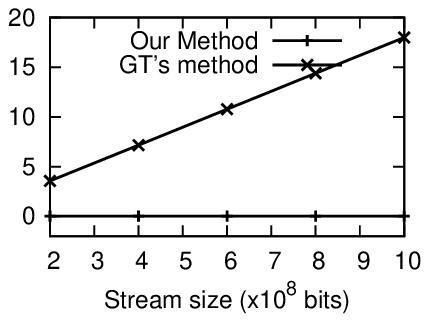}}%
\caption{Stream size vs.\ time, $\epsilon = 0.2$}
\label{fig:size-time-eps0.2}
\end{figure*}

\begin{figure*}[h!]
  \centering
  \subfigure[Audio Bible]{\label{fig:size-time-eps0.5-1} \includegraphics[scale=1.0]{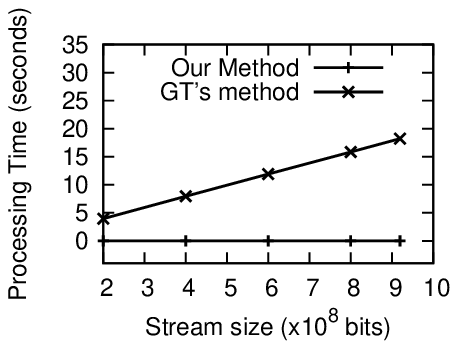}}%
  \subfigure[Video of President]{\label{fig:size-time-eps0.5-2} \includegraphics[scale=1.0]{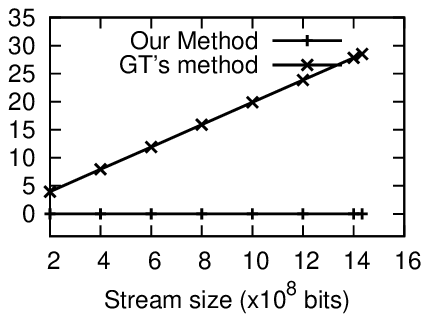}}%
  \subfigure[Earth Image]{\label{fig:size-time-eps0.5-3} \includegraphics[scale=1.0]{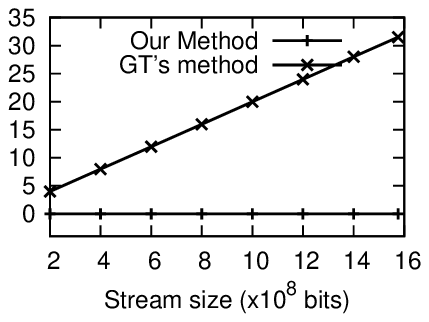}}%
  \newline
  \subfigure[Worldcup 98]{\label{fig:size-time-eps0.5-4} \includegraphics[scale=1.0]{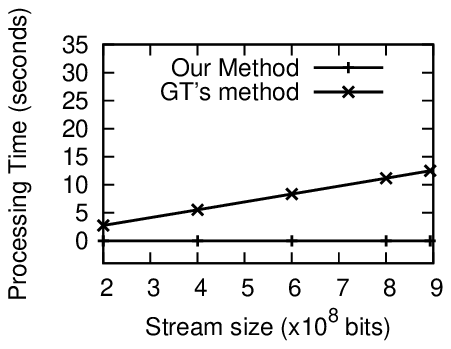}}%
  \subfigure[Synthetic-0.3]{\label{fig:size-time-eps0.5-5} \includegraphics[scale=1.0]{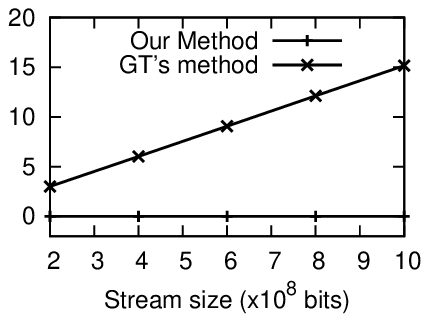}}%
  \subfigure[Synthetic-0.4]{\label{fig:size-time-eps0.5-6} \includegraphics[scale=1.0]{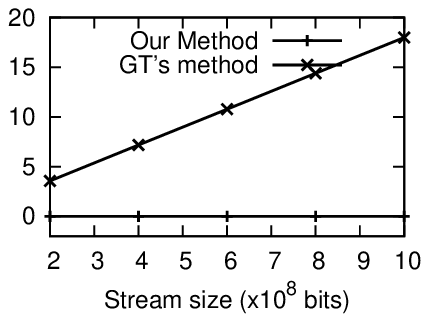}}%
\caption{Stream size vs.\ time, $\epsilon = 0.5$}
\label{fig:size-time-eps0.5}
\end{figure*}

\end{document}